\documentclass[12pt]{article}
\usepackage{amsfonts}
\usepackage{amssymb}
\usepackage{graphicx}
\usepackage{amsmath}
\usepackage{harvard}

\newtheorem{theorem}{Theorem}

\newtheorem{corollary}{Corollary}

\newtheorem{lemma}{Lemma}

\newtheorem{proposition}{Proposition}

\newenvironment{proof}[1][Proof]{\textbf{#1.} }{\ \rule{0.5em}{0.5em}}

\renewcommand{\cite}{\citeasnoun}
\begin{document}

\author{Dirk Bergemann\thanks{%
Yale University, dirk.bergemann@yale.edu} \and Tibor Heumann\thanks{%
Pontificia Universidad Cat\'{o}lica de Chile, tibor.heumann@uc.cl} \and %
Stephen Morris\thanks{%
Massachusetts Institute of Technology, semorris@mit.edu}}
\title{Procurement without Priors:\linebreak A Simple Mechanism and its
Notable Performance\thanks{%
We acknowledge financial support from NSF grants SES-2001208 and
SES-2049744, and ANID Fondecyt Regular 1241302. An Extended Abstract of an
early version of this paper appeared under the title "Cost-Based Nonlinear
Pricing" in the Conference Proceedings of ACM-EC '23. We thank seminar
audiences at Cambridge, Chicago, Essex, Gerzensee, LSE, Oxford, Stanford,
Texas, UBC, UCL, and UCLA for helpful discussions. We have benefitted from
comments by Jerry Anunrojwong, Costas Arkolakis, Omar Besbes, Ben Brooks,
Yang Cai, Eduardo Davila, Daniel Garrett, Jason Hartline, Yingkai Li, Zvika
Neeman, Xavier Vives and Jinzhao Wu. We thank Hongcheng Li and David Wambach
for excellent research assistance.} }
\date{December 8, 2025}
\maketitle

\begin{abstract}
How should a buyer design procurement mechanisms when suppliers' costs are
unknown, and the buyer does not have a prior belief? We demonstrate that
simple mechanisms -- that share a \emph{constant fraction} of the buyer
utility with the seller -- allow the buyer to realize a guaranteed positive
fraction of the efficient social surplus across all possible costs.
Moreover, a judicious choice of the share based on the known demand
maximizes the surplus ratio guarantee that can be attained across all
possible (arbitrarily complex and nonlinear) mechanisms and cost functions.
Similar results hold in related nonlinear pricing and optimal regulation
problems.

\medskip

\noindent \textsc{Jel Classification: }D44, D47, D83, D84.

\end{abstract}


\newpage

\section{Introduction\label{sec:into}}


\subsection{Motivation\label{subsec:mot}}

Procurement policies face a fundamental challenge in both public and private
sector operations. The buyer must offer a purchase mechanism without knowing
the supplier's true costs. A government agency awarding infrastructure
contracts, a firm purchasing components, or a regulator setting utility
prices---all must commit to payment rules with limited information about
production technologies, input costs, or operational efficiencies. Bayesian
mechanism design, pioneered by \cite{myer81}, offers a solution: if the
buyer has a Bayesian prior over suppliers' costs, optimal mechanisms can be
characterized precisely. But this prescription demands substantial
informational sophistication. In practice, buyers often lack reliable
probability distributions over cost structures, particularly when: (i)
procuring novel products or services, (ii) operating across heterogeneous
contexts (local, regional, federal levels), and (iii) facing rapidly
evolving technologies or input markets. This paper asks: What can a buyer
guarantee themselves when they know their own demand but possess minimal
information about suppliers' costs?

We propose a remarkably simple class of mechanisms---\emph{constant utility
share mechanisms}---that offer powerful performance guarantees. Under such
mechanisms, the buyer commits to paying the supplier a fixed share of the
buyer's gross utility for any quantity supplied. Despite knowing nothing
about the supplier's cost function, the buyer can guarantee themselves a
definite positive fraction of what could be achieved as the social welfare
under complete information. Thus, the buyer can guarantee himself, with a
simple mechanism and unknown supply, a constant fraction of the efficient
social surplus that could be attained with an optimal mechanism and complete
information about supply. Furthermore, a judicious choice of the share
(calibrated to the demand) can guarantee the maximum share of the efficient
social surplus across all, possibly arbitrarily complex, mechanisms.

The maximum share of the efficient social surplus under complete information
that the buyer can guarantee themselves is referred to as the \emph{%
competitive ratio} in theoretical computer science and operations research,
see \cite{elbo98}. Thus, we show that a simple mechanism attains the
competitive ratio for our problem. In other words, we establish that
constant share mechanisms maintain their performance guarantees across all
possible supplier cost functions; furthermore, they provide robust
procurement efficiency in the sense that they attain the competitive ratio.
\ We view this as a supporting argument in favor of the simple mechanisms we
study.

We first present our approach in the case where the supplier has a constant
marginal cost of production, and the buyer has a power utility function. \
The demand generated by any specific power utility function generates a
constant price elasticity of demand. We introduce constant share mechanisms
and identify a specific constant share mechanism that generates the maximum
surplus ratio for the buyer (Proposition \ref{prop:css}). \ We then show
that the maximum surplus ratio attained by a constant share mechanism is
also the maximum surplus ratio attained by \textit{any} mechanism and thus
attains the competitive ratio. \ We thus establish that constant share
mechanisms solve the max min problem, where the buyer chooses a mechanism to
maximize the surplus ratio, subject to nature choosing the distribution of
costs to minimize the surplus ratio (Proposition \ref{prop:ocss}). \ We can
provide a closed form for the optimal constant share mechanism and thus the
competitive ratio as a function of the demand elasticity: the competitive
ratio decreases from $1$ to $1/e$ as the magnitude of the demand elasticity
varies from $1$ to $\infty $. \ We then show that there exists a sequence of
probability distributions over constant marginal costs such that a constant
share mechanism is the limit of the solutions to the classical optimal
Bayesian mechanism for that sequence of beliefs. \ This allows us to
establish a saddle point property of the optimal constant share mechanism:\
the optimal constant share mechanism also solves the min max problem where
nature first chooses a distribution over (constant) marginal costs to
minimize the surplus ratio, subject to the buyer choosing an optimal
mechanism for that distribution (Proposition \ref{bmech}). \ 

All these results extend from the constant marginal cost case to our general
setting with weakly increasing and weakly convex cost functions. \ Theorem %
\ref{thm:osg} establishes the competitive ratio as a simple function of the
buyer's demand elasticity. As the magnitude of the demand elasticity becomes
larger, the share of the surplus conceded to the seller increases and the
competitive ratio decreases.\ We also show how our results can be extended
to allow variable elasticity of demand and non-convex costs, although now
the interpretation of results is more subtle, as the lower bound is not the
solution of the max min problem anymore (Theorem \ref{thm:rev}).

We present our results for procurement mechanisms, but our approach extends
to other screening problems. \ We describe the related environments of
regulation (as in \cite{bamy82}) and nonlinear pricing (as in \cite{muro78}
and \cite{mari84}). In the regulatory problem, the regulator assigns a
positive weight to the profit of the seller, but possibly smaller than the
weight given to the buyer surplus. Theorem \ref{thm:wel} confirms that the
constant share mechanism continues to perform well. \ To the extent that the
regulator integrates the profit of the seller into the weighted social
surplus, the objective of the regulator becomes closer to the benchmark, the
social welfare with equal weights given to the buyer and seller. As a
consequence, the competitive ratio improves and the optimal share of surplus
conceded to the seller increases as well.

In the nonlinear pricing environment, we focus on power cost functions. This
allows us to express the cost environment of the firm in terms of a single
cost elasticity parameter. Proposition \ref{prop:seller} shows that a
constant mark-up mechanism tailored to the cost elasticity attains the
competitive ratio. Thus, a \emph{cost-based} nonlinear pricing approach
solves the nonlinear pricing problem in the same way that a \emph{%
demand-based} pricing solves the procurement problem.

The competitive ratio, similar to other objectives such as regret
minimization or maximin utility, has emerged in the analysis of robust
decision-making in the absence of a Bayesian prior distribution. We view the
competitive ratio as relevant and appropriate in the current context for two
reasons: \ First, it is a way to obtain a quantitative insight into the
robust policies. We compare these alternative criteria in Section \ref%
{sec:maxmin} and identify the advantages of the competitive ratio in the
present setting. Similar to earlier influential contributions in the theory
of regulation, see \cite{roge03} and \cite{chsa07}, we view the competitive
ratio as a practical quantity that provides guidance in the absence of a
Bayesian prior distribution. Second, in procurement settings, a procuring
agency has to articulate a policy that applies to all procuring events,
whether of small or large scale, and thus requires a policy that is
invariant to the scale of the procurement. The scale independence of the
competitive ratio -- a well-known feature -- is of particular relevance for
the design of procurement and regulatory policy. Procurement policies, by
public or private procurement agencies, have to define policies that attain
good performance across a wide scale of procurement events. By comparing the
attained buyer surplus relative to the efficient social surplus across all
possible scales, the resulting policy is performing well across all scales.
A high-performing policy has to realize a substantial share of the gains
from trade at every level of trade. Thus, the mechanism cannot exclude
trades at low quantities, nor does it attempt to aggressively discount at
large quantities. \ While we focus much of our discussion on procurement and
regulation, the emphasis on scale invariance is equally relevant for optimal
selling mechanisms. By measuring the success of the selling policy against
the social surplus, the selling policy is measured against the possible
social surplus, thus the size of the total market. For a private seller, say
a start-up or a seller of a new product, the selling policy just focuses on
the attainable social surplus rather than the expected profit on the initial
expectation that may be biased downwards relative to the eventual size of
the market.

\subsection{Related Literature\label{subsec:rel}}

Our work connects to several strands of literature. First, we contribute to
the literature on simple versus optimal mechanisms, exemplified by \cite%
{roge03} and \cite{chsa07} who analyzed linear contracts in procurement
settings. While previous work has focused on comparing simple mechanisms to
optimal ones under specific distributional assumptions, we characterize
conditions under which simple mechanisms perform well across all possible
cost structures.

The regulation setting of \cite{bamy82} with constant marginal cost is
closely related to our setting and we obtain competitive ratios for the
model of \cite{bamy82}. \cite{roge03} offers an analysis of simple
contracts, allowing only for a binary menu of either a fixed price contract
or a cost-reimbursement contract (FCBR menu) within the framework of \cite%
{lati86} which contains elements of adverse selection and moral hazard. The
seller has a cost level that is private information (adverse selection) and
reduces the unit cost by an additional effort (moral hazard). \cite{roge03}
shows that the optimal simple menu can attain 3/4 of the unrestricted
optimal menu, which consists of a menu of a continuum of linear contracts in
a parametrized version of quadratic cost of effort and uniformly distributed
private information (Proposition 2 and Section 3 of \cite{roge03}). \cite%
{chsa07} provides a more complete analysis of the performance of simple
contracts by allowing a larger class of distributions and a larger class of
contracts. They obtain a competitive ratio of the simple relative to the
optimal regulation solution of $2/e$.

In the theory of regulation, a concern for robustness is inherent to the
problem. \cite{garr14} considers the model of \cite{lati86} and shows that
the simple menu proposed by \cite{roge03} is the optimal robust menu when
the regulator is uncertain about the cost reduction technology of the firm. 
\cite{garr21} pursues this line of research further and asks what an analyst
who is uncertain of the cost reduction technology can assert about the
payoff implications for the regulator and the firm when the regulator is
with knowledge about the technology offers an optimal incentive compatible
regulation.

\cite{gush25} consider the regret-minimizing regulation policy in the
setting of \cite{bamy82} when the regulator is uncertain about the demand
and supply for the product or services of a monopoly. They consider the
regret-minimizing policy rather than the competitive ratio.

We derive performance guarantees and robust pricing policies that secure
these guarantees for large classes of second-degree price discrimination
problems as introduced by \cite{muro78} and \cite{mari84}. We do this
without imposing any restriction on the distribution of the values, such as
regularity or monotonicity assumptions, or finite support conditions. We
only require that the social surplus of the allocation problem has a finite
expectation.

The optimal monopoly pricing problem for a single object is a special
limiting case of our model when the marginal cost of increasing supply
becomes infinitely large. The analysis of the competitive ratio in the
single-unit case is also a special case of a result of \cite{neem03}. He
investigates the performance of English auctions with and without reserve
prices. The case of the optimal monopoly pricing is a special case of the
auction environment with a single bidder. He establishes a tight bound for
the single-bidder case that is given by a \textquotedblleft truncated Pareto
distribution\textquotedblright\ (Theorem 5). The bound that he derives is a
function of a parameter that is given by the ratio of the bidder's expected
value and the bidder's maximal value. Without the introduction of this ratio
as a parameter, the bound is zero: as this ratio converges to zero, so does
the bound. Similarly, \cite{erma10} and \cite{haro14} establish that for
distributions with support $\left[ 1,h\right] $, the competitive ratio of
the optimal pricing problem is $1+\ln h$. \ Thus, as $h$ grows, the
approximation becomes arbitrarily weak (see Theorem 2.1 of \cite{haro14}).
By contrast, our results obtain a constant approximation for every convex
cost function. Thus, the introduction of a convex cost function (or a
concave utility function) leads to a noticeable strengthening of the
approximation quality.

\cite{cfmm19} consider a \ different robust version of the nonlinear
allocation problem. The seller faces a privately informed buyer and only
knows the first moment of the distribution and its finite support (taken to
be the unit interval). As in \cite{besc08} they solve the problem by
characterizing equilibria of an auxiliary zero-sum game played by the seller
and an adversarial nature. Their main result is that in the optimal robust
mechanism the ex-post payoff of the seller has to be linear in the buyer's
realized type. This property, which can be traced to the moment constraint
on the mean, does not hold for our solution. \cite{carr17} considers a
robust version of the multi-item pricing problem. The buyer has an additive
value for finitely many items and has private information about the value of
each item. There, the seller knows the marginal distribution for every item
but is uncertain about the joint value distribution. \cite{carr17} solves a
minmax problem and shows that separate item-by-item pricing is the robustly
optimal pricing policy. \cite{dero24} and \cite{chzh25} consider a related
problem under informational robustness. There, the joint distribution of
values is given by a commonly known distribution, but nature or the buyer
can choose the optimal information structure. In the corresponding solution
of the mechanism design problem, they show that the optimal mechanism is
always one of pure bundling. \cite{mipp25} considers a robust regulation
problem where the buyer is uncertain both about their own utility function
and the distribution over (constant) marginal cost of the seller. The buyer
is then following a two-step procedure. They first identify the set of
mechanisms that provides the highest welfare guarantee against the set of
all possible value and cost functions, and then in a second step choose from
the short list of mechanisms the one that maximizes the welfare against the
conjectured value and cost. The setting of uncertainty, as well as the
analysis, is motivated by ambiguity and misspecification concerns on both
sides of the market, and thus quite distinct from our perspective.

\section{Model\label{sec:model}}

\subsection{Payoffs}

A (large)\ buyer procures quantity $q\in \mathbb{R}_{+}$ of a product from a
seller. The buyer's utility gross of transfers is increasing and concave,
denoted by: 
\begin{equation}
u(q)=\frac{q^{\sigma }}{\sigma },  \label{u}
\end{equation}%
with $\sigma \in (0,1).$ The variable $q\in \mathbb{R}_{+}$\ can also be
interpreted as the quality of a product. The seller's cost of providing
quantity $q$ is given by a cost function: 
\begin{equation*}
c:\mathbb{R}_{+}\rightarrow \mathbb{R}_{+}.
\end{equation*}%
The cost function is assumed to be (weakly) increasing and (weakly) convex.
Buyer and seller have net utility and profit functions that are quasi-linear
in transfers $t$. We maintain these assumptions on the utility and cost
functions throughout the paper, except in Section \ref{subsec:ext} where we
examine the robustness of our results.

The buyer's demand $D\left( p\right) $ for the good at uniform price $p$ is
denoted by: 
\begin{equation*}
D\left( p\right) \triangleq \underset{q}{\arg \max }\left\{ u\left( q\right)
-pq\right\} .
\end{equation*}%
From the first-order condition, we obtain that this is given by: 
\begin{equation*}
D\left( p\right) =(u^{\prime })^{-1}(p),
\end{equation*}%
The power utility function means the buyer demand has a constant (negative)\
elasticity of demand: 
\begin{equation*}
\varepsilon _{D}\left( p\right) =-\frac{\frac{D\left( p\right) }{dp}}{\frac{%
D\left( p\right) }{p}}=\frac{1}{1-\sigma }\in \left( 1,\infty \right) \text{.%
}
\end{equation*}%
Analogously, the seller's supply function is 
\begin{equation*}
S(p)=(c^{\prime })^{-1}(p).
\end{equation*}%
That is, at per-unit price $p$ the seller's supplied quantity is $S(p)$. The
assumption that the cost is convex corresponds to an increasing supply.

\subsection{Private Information and Mechanism}

The seller has private information about their cost function.\ The buyer
does not know the seller's cost function. We denote by $\mathcal{C}$ the set
of feasible cost functions that the buyer considers possible. For example,
the class $\mathcal{C}$\ of feasible cost functions could consist of all
linear or all convex cost functions. For now, we only require that the cost
is not identical to $0$ so that the efficient social surplus is finite.

The buyer chooses a mechanism that elicits the private information of the
seller and determines the level of production $q$. \ By the taxation
principle (see Proposition 1, \cite{gula84}), it is without loss to consider
the indirect mechanism described by a nonlinear payment rule: 
\begin{equation*}
t:\mathbb{R}_{+}\rightarrow \mathbb{R}_{+}.
\end{equation*}%
We assume that $t(0)=0$, which guarantees that the indirect mechanism
satisfies the participation constraint. \ We write $T$ for the set of
mechanisms. \ 

Now the seller chooses a quantity $q$ depending on her cost function $c$ and
the mechanism $t$:%
\begin{equation*}
q\left( c,t\right) \triangleq \underset{q\in \mathbb{R}_{+}}{\arg \max }%
\left\{ t\left( q\right) -c\left( q\right) \right\} .
\end{equation*}%
If multiple maximizers exist, we assume that the seller chooses the smallest
one. Hence, the choice is adversarial from the buyer's perspective. The
precise way in which the seller breaks indifferences makes no substantive
difference for the analysis.


\subsection{The Buyer Surplus Ratio Guarantee}

The buyer does not know the true cost function, only knowing that the true
cost $c$ is within the class $\mathcal{\ C}$ of cost functions. In
particular, the buyer does not have any prior belief over cost functions.

In the absence of a prior distribution over feasible cost functions $%
\mathcal{C}$, the buyer cannot compute the \emph{expected} buyer surplus of
any given mechanism. In the absence of a prior distribution, we will compare
the realized buyer surplus with the social surplus under complete
information. The buyer, therefore, seeks to identify the mechanism $t$ that
guarantees them the largest possible share of the efficient social surplus
as buyer surplus. The efficient social surplus is given by the maximum
attainable utility net of cost under complete information: 
\begin{equation*}
W\left( c\right) \triangleq \ \underset{q\in \mathbb{R}_{+}}{\max }\left\{
u\left( q\right) -c\left( q\right) \right\} .
\end{equation*}%
The efficient social surplus depends only on the realized cost function $c$
but not the mechanism $t$. The efficient quantity supplied is denoted by: 
\begin{equation*}
q^{\ast }\left( c\right) \triangleq \ \underset{q\in \mathbb{R}_{+}}{\arg
\max }\left\{ u\left( q\right) -c\left( q\right) \right\} ,
\end{equation*}%
which will be uniquely defined.

The buyer surplus $U$\ is determined by the supply $q\left( c,t\right) $ and
the payment $t\left( q\right) $\ for the supplied quantity $q$: 
\begin{equation}
U\left( c,t\right) \triangleq u\left( q\left( c,t\right) \right) -t\left(
q\left( c,t\right) \right) .  \label{cs0}
\end{equation}%
The seller surplus, or profit, $\Pi $ is given by 
\begin{equation}
\Pi \left( c,t\right) \triangleq t\left( q\left( c,t\right) \right) -c\left(
q(c,t)\right) .  \label{eq:ps}
\end{equation}%
We evaluate the performance of a mechanism $t$\ using the buyer surplus
ratio guarantee: 
\begin{equation*}
\min_{c\in \mathcal{C}}\frac{U(c,t)}{W(c)}.
\end{equation*}%
If there was no private information regarding the cost $c$, then the buyer
would be able to extract the full efficient social surplus $W(c)$ and the
buyer surplus ratio would attain its maximum possible value of $1$. The
buyer surplus ratio guarantee is the proportion of the efficient social
surplus that the mechanism is guaranteed to provide. \ 

\ The competitive ratio is the maximum possible buyer surplus ratio
guarantee: 
\begin{equation}
\max_{t\in T}\min_{c\in \mathcal{C}}\frac{U\left( c,t\right) }{W\left(
c\right) }.  \label{ratio}
\end{equation}%
A mechanism $t$\ is optimal if it attains the competitive ratio. \ A
strictly positive competitive ratio offers the buyer a guarantee that is
proportional to the efficient social surplus, irrespective of the scale of
the problem. In Section \ref{sec:maxmin}, we provide a detailed discussion
relating the competitive ratio to other decision criteria used in the
presence of non-Bayesian uncertainty.

The $\max \min $ competitive ratio has an associated $\min \max $ that is
weakly larger: 
\begin{equation}
\max_{t\in T}\min_{c\in \mathcal{C}}\frac{U\left( c,t\right) }{W\left(
c\right) }\leq \min_{c\in \mathcal{C}}\max_{t\in T}\frac{U\left( c,t\right) 
}{W\left( c\right) }.  \label{eq:saddle}
\end{equation}%
With deterministic price schedules and costs, there is typically a gap
between $\max \min $ and $\min \max $, but when there is no gap, often
requiring stochastic strategies, then we refer to the resulting solution as
a \emph{saddle point}. After we provide Theorem \ref{thm:osg} we discuss the
performance of stochastic mechanisms.

\section{The Linear Cost Environment\label{sec:lc}}

We first illustrate our results with the class of linear cost functions. \
The seller produces quantity $q\in \mathbb{R}_{+}$ with a linear cost
function%
\begin{equation*}
c\left( q\right) =c\cdot q,\ \ c\in \mathbb{R}_{+}.
\end{equation*}%
We abuse notation by writing $c$ for the constant marginal cost in this
Section. The constant marginal cost $c\in \mathbb{R}_{+}$ is known to the
seller, but unknown to the buyer. In this environment, the relevant
properties of the buyer and seller are summarized in a single parameter
each, the utility exponent $\sigma $ of the utility function (see \eqref{u})
and the marginal cost $c$, respectively.

We proceed as follows in this Section. \ We first restrict attention to 
\emph{constant share mechanisms}, where the transfer is a constant fraction
of the buyer's utility. \ We show that this class of mechanisms gives rise
to a positive ratio guarantee. We then show that the ratio guarantee from
the constant share mechanism is the maximum that can be attained by \emph{any%
} mechanism, and thus attains the competitive ratio. Finally, we exhibit a
specific class of distributions over the cost parameter $c$, and show that
the constant share mechanism can be rationalized as the Bayes optimal
mechanism against this distribution. As a by-product, we prove the existence
of a saddle point. \ This provides an indirect proof that a constant share
mechanism attains the competitive ratio. Hence, we provide two distinct
arguments for the optimality of constant share mechanisms. The direct
argument is closer to the arguments we use in the general case, so it is a
better illustration of the general arguments we later provide, even if, for
this case, we have a shorter (indirect) proof using a specific Bayesian
distribution over cost.

The main takeaway from this section is that constant share mechanisms
perform well. \ We also provide strong economic intuitions for why they
perform well.

\subsection{Constant Share Mechanisms}

We first ask whether a constant surplus share mechanism allows the buyer to
realize a significant buyer surplus in the absence of information about the
true cost of the seller. \ A constant share mechanism takes the form 
\begin{equation}
t\left( q\right) =z\cdot u\left( q\right) \text{, }z\in \left( 0,1\right) ,
\label{eq:cs}
\end{equation}%
so the buyer compensates the seller with a constant fraction $z\in \left(
0,1\right) $\ of his utility. Under this constant share mechanism, the buyer
retains $\left( 1-z\right) u\left( q\right) $ as their buyer surplus. \ 

The seller does not know that the mechanism is a constant share mechanism
(i.e., the seller knows $t$ but does not care about $z$ and $u$). \ But the
given the constant share mechanism, we know that seller surplus will equal : 
\begin{equation}
\max_{q\in \mathbb{R}_{+}}\left\{ \frac{z}{^{\sigma }}q^{\sigma }-cq\right\}
,  \label{eq:obj}
\end{equation}%
Thus the quantity supplied as a function of the seller's realized cost $c$
and the share $z$ is 
\begin{equation}
q\left( c,z\right) =\left( \frac{z}{c}\right) ^{\frac{1}{1-\sigma }}.
\label{eq:sup}
\end{equation}%
The optimal quantity supplied $q\left( c,z\right) $ is decreasing in the
marginal cost $c$ and increasing in the seller share $z$. The resulting
buyer surplus is%
\begin{equation}
U\left( c,z\right) =\frac{1-z}{\sigma }\left( \frac{z}{c}\right) ^{\frac{%
\sigma }{1-\sigma }}.  \label{eq:csur}
\end{equation}%
The constant-share mechanism has several noteworthy features. It guarantees
that the seller with any finite marginal cost makes a sale, as $q\left(
c,z\right) >0$ for all $c\in \mathbb{R}_{+}$. \ Thus the mechanism (\ref%
{eq:cs}) does not exclude any seller. Moreover, the supply of $q$ is
increasing in the share $z$\ of the surplus conceded as indicated by (\ref%
{eq:sup}). The seller surplus is given by: 
\begin{equation*}
\Pi \left( c,z\right) =\frac{z\left( 1-\sigma \right) }{\sigma }\left( \frac{%
z}{c}\right) ^{\frac{\sigma }{1-\sigma }}\text{.}
\end{equation*}

The expression for buyer surplus (\ref{eq:csur}) indicates that there is a
clear trade-off in the choice of the optimal share $z$. A larger share $z$
results in larger supply of $q$ but the buyer surplus is a smaller share $%
1-z $ of the gross buyer utility. The share $z$ that maximizes the buyer
surplus ratio for any realized cost $c$ is given by: 
\begin{equation}
z^{\ast }=\sigma \text{,}  \label{eq:zstar}
\end{equation}%
that is, the optimal share $z^{\ast }$ matches the exponent $\sigma $\ of
the power utility function.

Thus, if we \emph{were to restrict} attention to constant share mechanisms,
the optimal share for the buyer is given by (\ref{eq:zstar}). In other
words, for any realization of $c$, the largest (absolute) buyer surplus is
attained by $z^{\ast }=\sigma $. But of course, for any realized cost $c$,
there exists a non-constant mechanism that could guarantee a higher buyer
surplus. If the buyer knew the realization of $c$, he could offer a
compensation equal to the cost of providing the socially efficient quantity,
as long as this quantity is supplied, and zero otherwise. This is the
outcome that would prevail in \emph{first-degree price discrimination}.

But in the absence of the information about the realized cost $c$ or any
distributional information about the cost, it is difficult to say how well
the buyer does by adopting the best constant sharing rule $z^{\ast }=\sigma $
relative to any other possible non-constant share mechanism.

\subsection{The Ratio Guarantee with Constant Share Mechanisms \label%
{subs:cr}}

We assume the buyer compares the buyer surplus with the social surplus under
complete information. Social surplus is maximized by quantity 
\begin{equation*}
q^{\ast }\left( c\right) =\left( \frac{1}{c}\right) ^{\frac{1}{1-\sigma }},
\end{equation*}%
and the efficient social surplus is given by: 
\begin{equation}
W\left( c\right) =\frac{1-\sigma }{\sigma }\left( \frac{1}{c}\right) ^{\frac{%
\sigma }{1-\sigma }}.  \label{eq:socs}
\end{equation}%
Now the buyer surplus ratio of (\ref{eq:csur}) and (\ref{eq:socs}),
respectively, is:%
\begin{equation}
\frac{U\left( c,z\right) }{W\left( c\right) }=\frac{\frac{1-z}{\sigma }%
\left( \frac{z}{c}\right) ^{\frac{\sigma }{1-\sigma }}}{\frac{1-\sigma }{%
\sigma }\left( \frac{1}{c}\right) ^{\frac{\sigma }{1-\sigma }}}=\frac{1-z}{%
\left( 1-\sigma \right) z^{\frac{\sigma }{\sigma -1}}}.  \label{eq:ratl}
\end{equation}%
This ratio does not depend on the cost realization $c$. Therefore, we can
find a sharing rule $z^{\ast }$ which maximizes the ratio and only depends
on the exponent $\sigma $ of the utility function. We summarize these
findings in the proposition below.

\begin{proposition}[Constant Share Mechanisms]
\label{prop:css}$\quad $\newline
For all constant share mechanisms, $t(q)=z\cdot u\left( q\right) $, we have
that 
\begin{equation*}
\min_{c\in \mathbb{R}_{+}}\frac{U(c,z)}{W(c)}\leq \sigma ^{\frac{\sigma }{%
1-\sigma }}.
\end{equation*}%
The inequality is an equality if and only if $z^{\ast }=\sigma $, in which
case the bound is attained for all $c\in \mathbb{R}_{+}$. 
\end{proposition}

We thus find that the ratio guarantee for the optimal constant share
mechanism. \ The buyer surplus share is bounded below by $1/e\approx 0.37$
as $\sigma ^{\frac{\sigma }{1-\sigma }}$ is decreasing in $\sigma $ and
converges to $1/e$ as $\sigma \rightarrow 1$. In this limit, the demand
becomes perfectly elastic. This is illustrated in Figure~\ref{fig11}. While
in this setting the ratio guarantee is at least $1/e$, we will see that the
ratio guarantee converges to 0 as the $\sigma$ converges to 1 when the buyer
considers all convex cost functions.

In fact, the constant share mechanism also provides a remarkable performance
in terms of efficiency. The \emph{joint surplus} that buyer and seller can
achieve with this constant share mechanism is bounded below by $2/e\approx
0.74$, thus leaving buyer and seller with a substantial part of the social
surplus despite the inefficiency due to private information. In Figure \ref%
{fig22} we display the joint behavior of buyer and seller surplus.

\begin{figure}[t]
	\centering
	\includegraphics[width=5.418in,height=2.5374in,keepaspectratio]{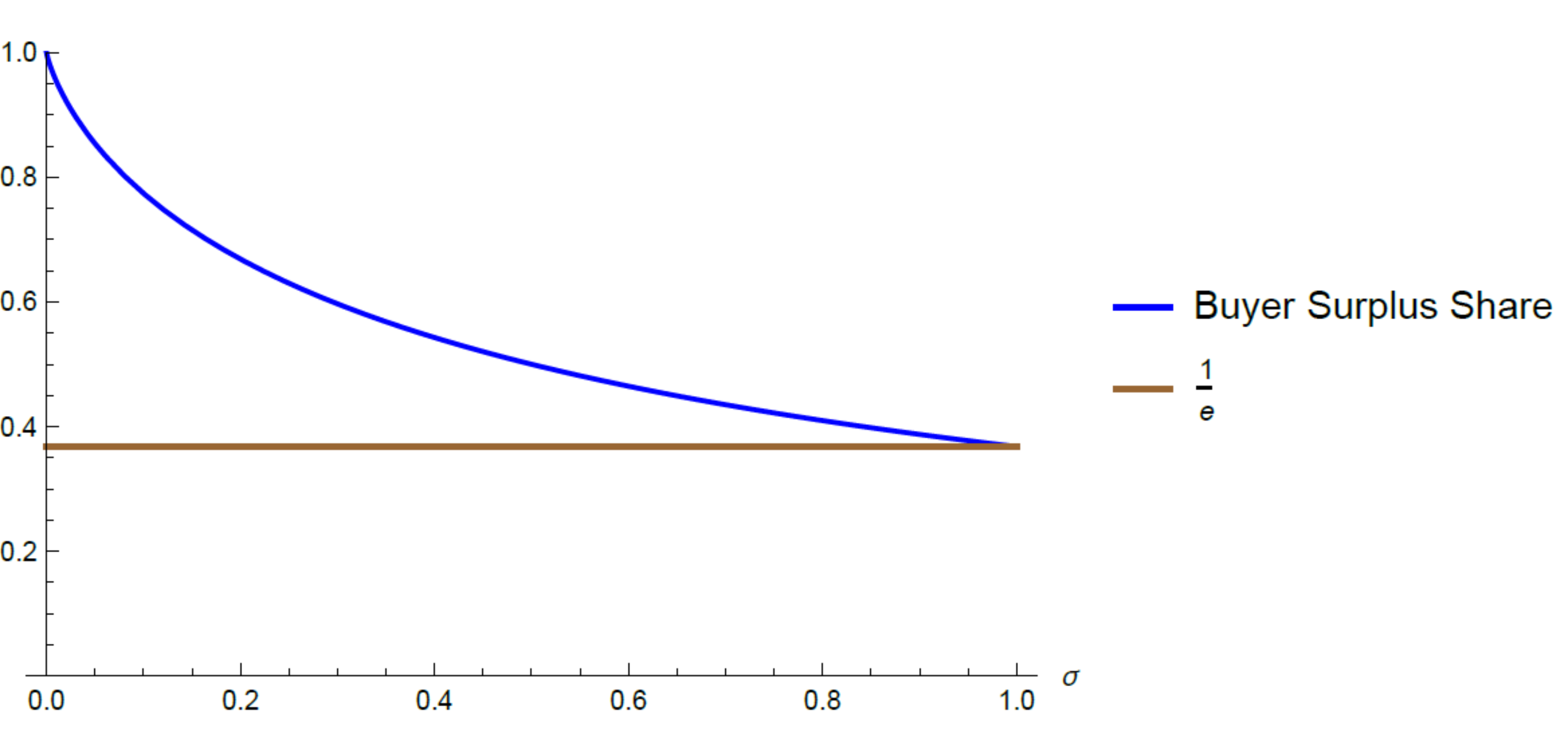}
	\caption{Buyer Surplus Guarantee as Share of Social Surplus}
	\label{fig11}
\end{figure}

\begin{figure}[t]
	\centering
	\includegraphics[width=5.6386in,height=2.5374in,keepaspectratio]{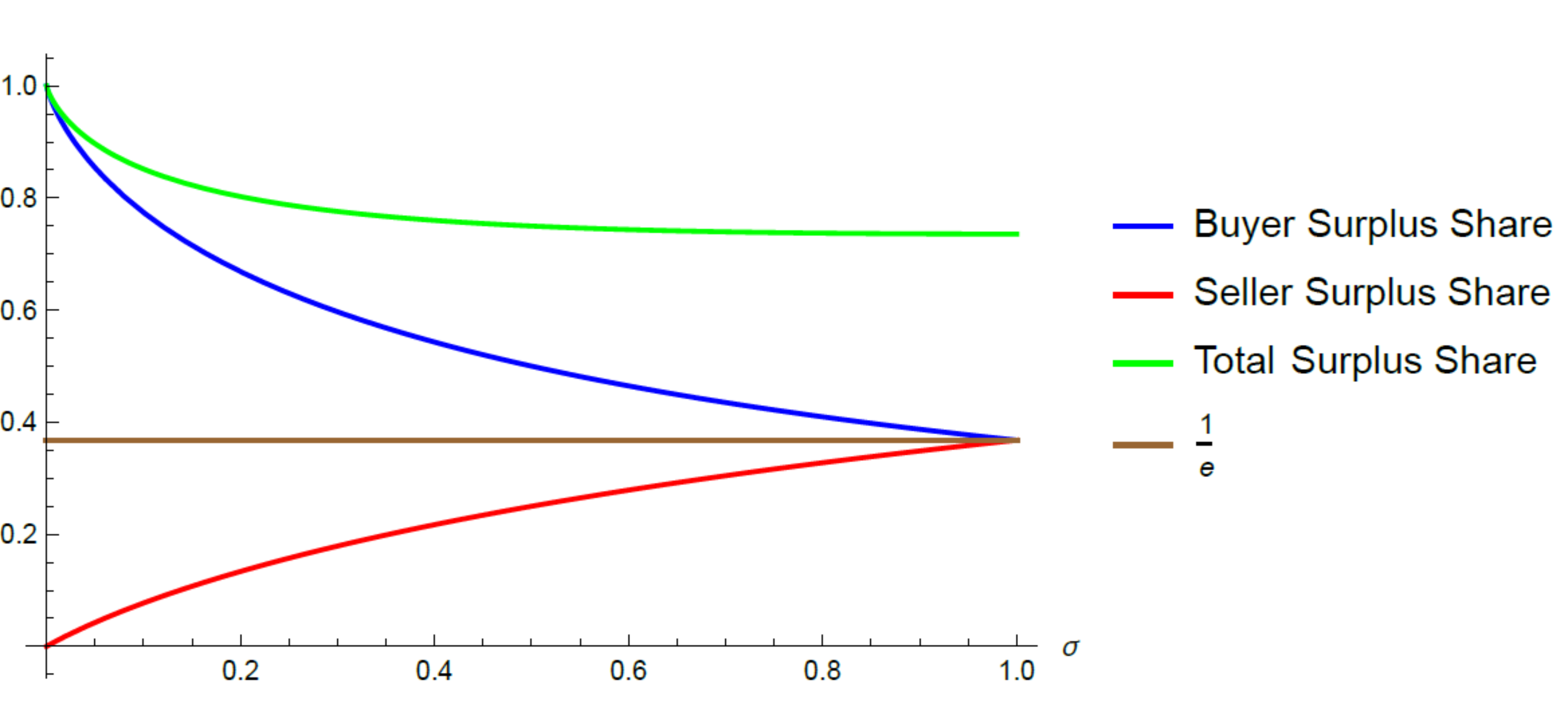}
	\caption{Buyer and Seller Surplus Guarantee as Share of Social Surplus}
	\label{fig22}
\end{figure}

\subsection{The Competitive Ratio and General Mechanisms}

In this previous section, we restricted attention to constant share
mechanisms. Within this class, we identified a particular sharing rule $%
z^{\ast }=\sigma $ that maximizes the absolute buyer surplus and thus the
buyer surplus ratio relative to the complete information social surplus. We
now consider all possible indirect mechanisms. We show that even in this
unrestricted class of mechanisms, the constant share mechanism $z^{\ast
}=\sigma $ maximizes the ratio guarantee and thus attains the competitive
ratio.

\newpage

\begin{proposition}[Competitive Ratio]
\label{prop:ocss}$\quad $\newline
For all mechanisms $t$, we have that: 
\begin{equation}
\min_{c\in \mathbb{R}_{+}}\frac{U(c,t)}{W(c)}\leq \sigma ^{\frac{\sigma }{%
1-\sigma }}.  \label{maxminex}
\end{equation}%
The inequality is tight if and only if $t\left( q\right) =z^{\ast }u\left(
q\right) $. 
\end{proposition}

We provide a sketch of this result (a rigorous proof is given in Proposition %
\ref{bmech} following a different route). That is, we prove that, for any $%
t(q)$ such that $t(q)\neq \sigma u\left( q\right) $ for some $q$, satisfies: 
\begin{equation*}
\min_{c\in \mathbb{R}_{+}}\frac{U(c,t)}{W(c)}<\sigma ^{\frac{\sigma }{%
1-\sigma }}
\end{equation*}%
To simplify the exposition, we consider only the case that the transfer $%
t(q) $ is concave, and the share of surplus conceded to the seller 
\begin{equation*}
z(q)\triangleq \frac{t(q)}{u(q)}
\end{equation*}%
is strictly decreasing at some $q=\widehat{q}$. (The argument we provide
next mirrors the central argument used to prove Theorem \ref{thm:osg}, which
we later provide. There we also consider the cases omitted in current
argument and allow for general nonlinear cost functions).

To make the notation more compact, we define: 
\begin{equation*}
\widehat{z}\triangleq z(\widehat{q})\text{ and }\widehat{c}\triangleq
t^{\prime }(\widehat{q}).
\end{equation*}%
Hence, by construction, when the cost is $c(q)=\widehat{c}q$, we have that: 
\begin{equation*}
\widehat{q}=\arg \max \{t(q)-c(q)\}
\end{equation*}%
We obtain this because by assumption $t(q)$ is concave so the objective
function is concave, and by construction $\widehat{q}$ satisfies the
first-order condition (that is, $c^{\prime }(\widehat{q})=t^{\prime }(%
\widehat{q})$). That is, $\widehat{z}$ is the share of surplus shared by the
buyer if the quantity sold is $\widehat{q}$ and $\widehat{c}$ is the
marginal cost that rationalizes $\widehat{q}$ as the optimal supply of
quality.

We now note that $z(q)$ is strictly decreasing at $\widehat{q},$ thus $%
z^{\prime }\left( \widehat{q}\right) <0$, and so by construction: 
\begin{equation*}
t^{\prime }(\widehat{q})=\widehat{z}u^{\prime }(\widehat{q})+z^{\prime
}\left( \widehat{q}\right) u(\widehat{q})<\widehat{z}u^{\prime }(\widehat{q}%
).
\end{equation*}%
This, in turn, implies that: 
\begin{equation}
q(\widehat{c},t)<q(\widehat{c},\widehat{z}).  \label{dcd}
\end{equation}%
In other words, under transfer $t$ and cost $c(q)=\widehat{c}q$, the supply
is smaller under transfer $t(q)$ than under the constant share mechanism $%
\widehat{t}(q)=\widehat{z}u\left( q\right) $. We thus have that: 
\begin{equation*}
\frac{U(\widehat{c},t)}{W(\widehat{c})}<\frac{U(\widehat{c},\widehat{z})}{W(%
\widehat{c})}\leq \frac{U(\widehat{c},z^{\ast })}{W(\widehat{c})}=\sigma ^{%
\frac{\sigma }{1-\sigma }}.
\end{equation*}%
The strict inequality follows from \eqref{dcd} while the weak inequality
follows from the definition of $z^{\ast }$ (note that the denominator does
not depend on $\widehat{z}$). Finally, the equality follows from Proposition~%
\ref{prop:css}.

This proves the result for the specific case we set out to prove. We assumed
that $z(q)$ is decreasing at some $q$; when this is not satisfied one can
repeat a similar argument by considering the limits of $q\rightarrow
\{0,\infty \}.$ Addressing the case in which $t(q)$ is not concave is
technically more cumbersome, but the arguments based on concavification
remain very similar. Overall, the intuition is that whenever the buyer
departs from a constant share mechanism, we can find a cost realization $%
\widehat{c}$ in which the seller supplies a quantity smaller than $q(c,%
\widehat{z}).$

\subsection{ The Constant Share Mechanism is a Bayes Optimal Mechanism}

We have shown that the constant sharing mechanism attains the highest \
guarantee among all feasible mechanisms. We now ask whether the constant
share mechanism may, in fact, constitute an optimal solution for a common
prior over the cost function. We now propose a specific distribution over
linear cost functions. We find that the Bayes optimal rule for this
distribution is again the simple sharing rule $z^{\ast }=\sigma $. With the
existence of the Bayes solution, we can then establish an equivalence
between the maximin and the minimax problem as described earlier in (\ref%
{eq:saddle}).

We consider a power distribution over the marginal cost $c$: 
\begin{equation*}
F(c)=\left( \frac{c}{\overline{c}}\right) ^{\alpha },
\end{equation*}%
with finite support $c\in \lbrack 0,\overline{c}]$ and a positive exponent $%
\alpha >0$. We first note that the expected efficient social welfare is
finite if and only if $\alpha >\sigma /(1-\sigma )$. That is: 
\begin{equation*}
\int W(c)dF(c)<\infty ,
\end{equation*}%
if and only if $\alpha >\sigma /(1-\sigma )$. We denote the Bayes-optimal
mechanism under the power distribution with exponent $\alpha $ by: 
\begin{equation*}
t_{\alpha }\triangleq \max_{t}\int U(c,t)dF(c).
\end{equation*}%
In the procurement problem, the virtual utility determines the optimal
allocation pointwise: 
\begin{equation*}
u^{\prime }\left( q\right) -\frac{F\left( c\right) }{f\left( c\right) }-c=0.
\end{equation*}%
It is then possible to compute the optimal mechanism $t_{\alpha }\left(
q\right) $ as the indirect mechanism from incentive and participation
constraints. In particular, we find that: 
\begin{equation*}
t_{\alpha }(q)=\frac{\alpha }{1+\alpha }u(q)-\frac{1-\sigma }{\sigma }\left( 
\frac{\alpha }{1+\alpha }\right) ^{\frac{1}{1-\sigma }}\left( \frac{1}{%
\overline{c}}\right) ^{\frac{\sigma }{1-\sigma }}.
\end{equation*}%
Hence, the Bayes optimal transfer is a constant share mechanism minus a
constant that is constructed to leave a seller with marginal cost $\overline{%
c}$ with zero transfer (and sellers with higher cost would be excluded, but
they have zero probability under the considered distribution). As we take
the limit $\overline{c}\rightarrow \infty $, the constant converges to zero
and the mechanism converges pointwise to a constant share mechanism.

\begin{proposition}[Bayes Optimal Mechanism]
\label{bmech}$\quad $\newline
The ratio of buyer surplus to social welfare attained by the Bayes optimal
mechanism satisfies: 
\begin{equation}
\lim_{\alpha \downarrow \frac{\sigma }{1-\sigma }}\frac{\int U(c,t_{\alpha
})dF(c)}{\int W(c)dF(c)}=\sigma ^{\frac{\sigma }{1-\sigma }}.  \label{fcvd}
\end{equation}%
Furthermore, the optimal mechanism satisfies: 
\begin{equation*}
\lim_{\substack{ \alpha \downarrow \frac{\sigma }{1-\sigma }  \\ \overline{c}%
\rightarrow \infty }}t_{\alpha }(q)=z^{\ast }u(q).
\end{equation*}
\end{proposition}

We thus obtain that the optimal constant share mechanism is indeed the Bayes
optimal mechanism for some distribution over marginal cost (by taking the
appropriate limits). Furthermore, the ratio of the efficient social surplus
attained by the Bayes optimal rule is the same as that obtained by the
optimal constant share mechanism. Now, we know that in general, the solution
of the above $\max \min $ problem is (weakly) below the solution of the
corresponding $\min \max $ problem:%
\begin{equation*}
\max_{t}\min_{F\in \Delta \mathbb{R}_{+}}\frac{U(F,t)}{W(F)}\leq \min_{F\in
\Delta \mathbb{R}_{+}}\max_{t}\frac{U(F,t)}{W(F)}.
\end{equation*}%
Hence, \eqref{maxminex} follows as an immediate implication from \eqref{fcvd}%
. One can then conclude the proof of Proposition \ref{prop:ocss} by noting
that \eqref{maxminex} is tight when the mechanism is a constant share
mechanism with sharing constant $z^{\ast }$ (Proposition \ref{prop:css}).

\section{The Nonlinear Cost Environment\label{sec:mini}}

We now analyze the performance of constant share mechanisms in an
environment with general nonlinear cost functions. We first consider the set
of all (weakly) convex cost functions and show that the constant share
mechanism guarantees the maximum possible share of the efficient social
surplus-- that is, a constant share mechanism attains the competitive ratio.
Hence, we provide the natural extension of Proposition \ref{prop:ocss}. In
this general environment, a saddle point does not exist in deterministic
strategies; that is, Proposition \ref{bmech} does not extend to this
setting. Yet, if we allow the buyer to randomize over mechanisms, we can
find a saddle point; we relegate the details of this generalization of
Proposition \ref{bmech} to the appendix (as Proposition \ref{prop:saddle}),
and discuss some details of the saddle point at the end of Section \ref%
{subsec:simple}.

In Section \ref{subsec:ext}, we provide the most general version of our
model: we relax the assumption of iso-elastic demand and convex cost. In
this more general setting, we show that constant share mechanisms still
guarantee a positive ratio guarantee. However, we do not show that a
constant share mechanism can attain the competitive ratio. Hence, while in
this more general setting our results are significantly weaker, we can still
establish that constant share mechanisms have desirable properties.

\subsection{The Environment with Convex Cost\label{subsec:simple}}

We denote by $\mathcal{C}_{cx}$ the set of all convex functions: 
\begin{equation*}
\mathcal{C}_{cx}\triangleq \{c:\text{ }c\left( q\right) \text{ is increasing
and convex in }q\}.
\end{equation*}%
Unless we explicitly state otherwise, all notions of monotonicity and
convexity are weak rather than strict. Throughout this subsection, we assume
that the buyer considers all convex cost functions to be plausible, that is, 
$\mathcal{C}=\mathcal{C}_{cx}$.

Given the large class\ of possible cost functions that the buyer faces,
namely $C_{cx}$, it will be critical to determine which cost functions may
minimize the ratio guarantee. Given a mechanism $t$\ and a cost function,
the quantity supplied by the seller $q(c,t)$ is determined by the optimality
condition that marginal revenue equals marginal cost: 
\begin{equation*}
t^{\prime }\left( q(c,t)\right) =c^{\prime }\left( q(c,t)\right) .
\end{equation*}%
The denominator of the ratio guarantee is the social welfare $W\left(
c\right) $ and hence to compute this ratio, the whole cost function is
relevant, and not just the marginal cost at the margin (i.e., $c^{\prime
}(q(c,t))$). For any given mechanism $t$\ we find the lowest convex cost $%
\widehat{c}$ that is consistent with keeping the quantity supplied
unchanged. That is, we find: 
\begin{equation}
\widehat{c}\triangleq \underset{\widetilde{c}\in \mathcal{C}_{cx}}{\arg\max}
\ W(\widetilde{c})\text{ \ \ \ \ \ \ \ \ \ \ subject to: }q\left( c,t\right)
=q\left( \widetilde{c},t\right) .  \label{max:c}
\end{equation}%
To make the notation more compact, we define the equilibrium choices under $%
c $ and $t$: 
\begin{equation*}
\widehat{q}\triangleq q\left( c,t\right) \text{ and }\widehat{\gamma }%
\triangleq c^{\prime }(\widehat{q}).
\end{equation*}%
That is, $\widehat{q}$ is the targeted quantity supplied under $\widehat{c}$
and $\widehat{\gamma }$ is the marginal cost at the quantity supplied. Given
the quantity supplied, the mechanism, and the cost, we can graphically find
the buyer surplus, profits, and deadweight loss, as illustrated in Figure %
\ref{fig00}. In the figure, we plot the marginal utility, marginal transfer,
and marginal cost, so the buyer surplus, profits, and deadweight loss
correspond to the appropriate areas between the curves.

The solution to \eqref{max:c} is: 
\begin{equation}
\widehat{c}(q)\triangleq 
\begin{cases}
0, & \text{if }q<\widehat{q}; \\ 
\widehat{\gamma }(q-\widehat{q}), & \text{if }q\geq \widehat{q};%
\end{cases}
\label{eq:plc}
\end{equation}%
which is the lowest convex cost that is consistent with the quantity
supplied being $\widehat{q}=q(c,t)$. That is, the marginal cost of providing 
$q\leq \widehat{q}$ is zero (and hence so is the total cost), and afterwards
it increases at constant rate $\widehat{\gamma }$. Here the marginal cost at
quantities $q\geq q(c,t)$ are disciplined by the fact that the cost is
convex, so the marginal cost cannot be smaller than $\widehat{\gamma }.$ In
Figure \ref{fig0} we the critical cost function and the additional profits
and deadweight loss relative to $c$. Note that by construction the supply is
the same as with cost $c$, so buyer surplus stays the same. Hence, the set
of critical cost functions that minimize the competitive ratio will consist
of piece-wise linear cost functions with a single kink at some $\widehat{q}$
with zero cost of providing $q\leq \widehat{q}$ and marginal cost $t^{\prime
}(\widehat{q})$ of providing $q\geq \widehat{q}$ .


\begin{figure}[t]
	\centering
	\includegraphics[width=6.4591in,height=2.7256in,keepaspectratio]{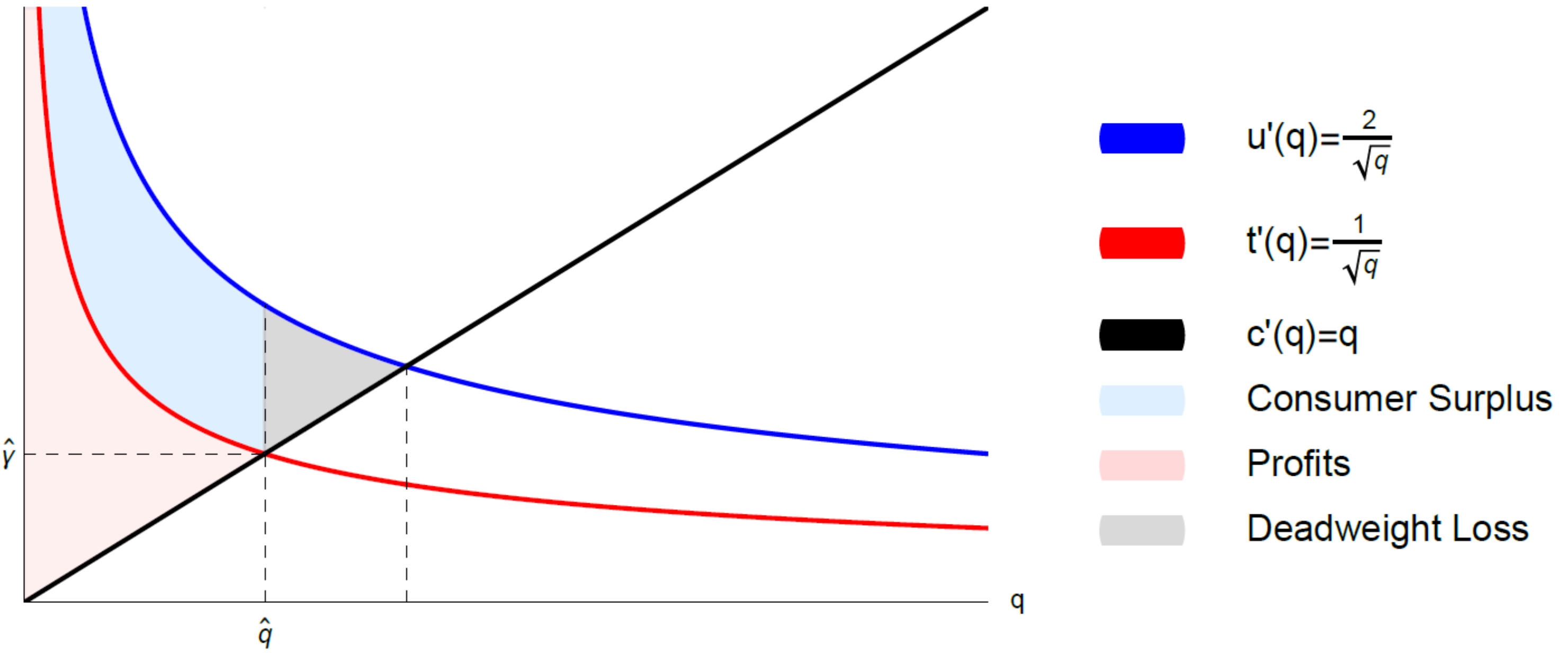}
	\caption{Quadratic Cost Function and Piecewise Linear Cost Function}
	\label{fig00}
\end{figure}

%
%

With this conjecture regarding the critical cost functions, for now, argued
informally but eventually proven formally, we can now construct a bound
based on a constant share mechanism. For any constant share mechanism 
\begin{equation*}
t(q)=z\cdot u(q),
\end{equation*}%
the ratio guarantee can be computed with the above piecewise linear cost
function, similar to the case of the linear cost function in Section \ref%
{subs:cr}. It is given by: 
\begin{equation*}
\frac{U\left( c,t\right) }{W\left( c\right) }=\frac{1-z}{(1-\sigma )z^{\frac{%
\sigma }{\sigma -1}}+\sigma z}.
\end{equation*}%
Relative to the ratio earlier (\ref{eq:ratl}), it contains an additional
term in the denominator, which reflects the zero cost of the initial units,
which increases the efficient social surplus.

\begin{figure}[t]
	\centering
	\includegraphics[width=6.8418in,height=2.7256in,keepaspectratio]{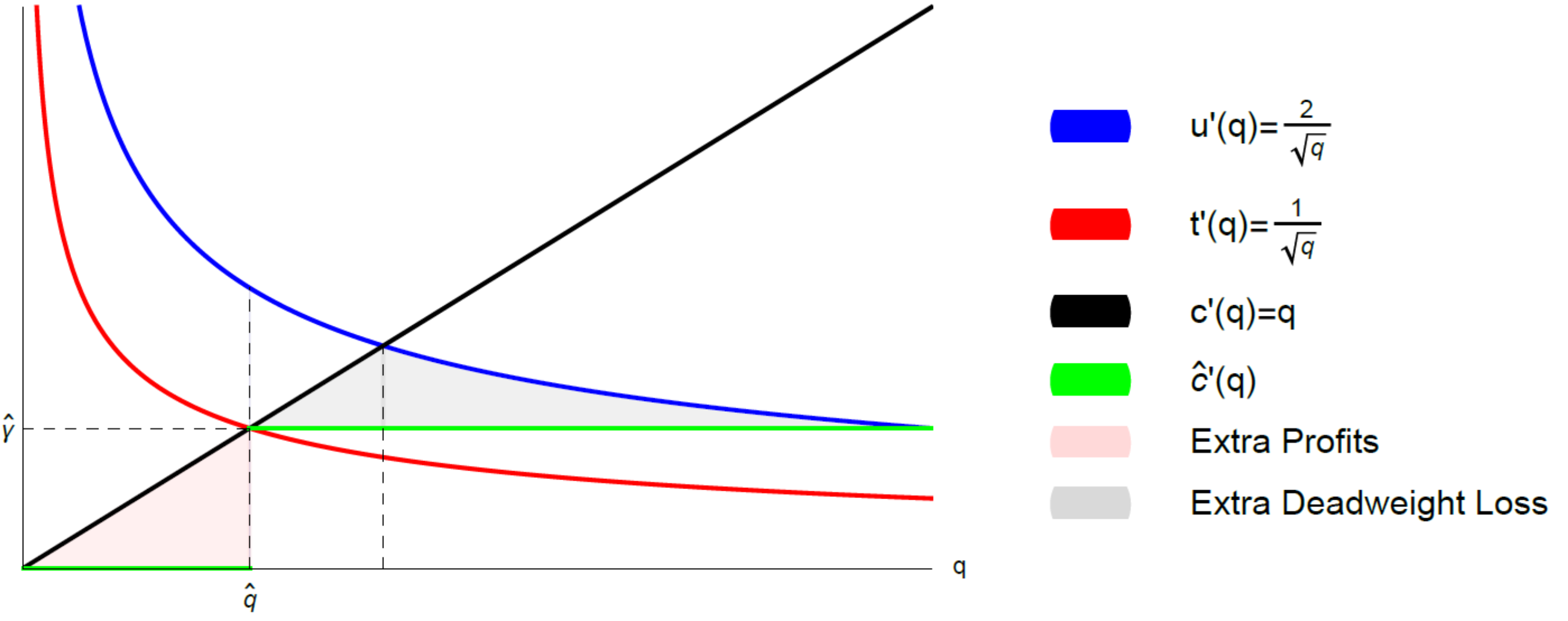}
	\caption{Buyers' Surplus and Social Surplus}
	\label{fig0}
\end{figure}

This leads us now to consider the following mechanism: 
\begin{equation}
t(q)=z^{\ast }(\sigma )u(q),  \label{eq:ss}
\end{equation}%
where $z^{\ast }(\sigma )$ solves the problem of maximizing the competitive
ratio under the specific piecewise linear cost function given by (\ref%
{eq:plc}): 
\begin{equation}
z^{\ast }(\sigma )\triangleq \arg \max_{z\in \lbrack 0,1]}\frac{1-z}{%
(1-\sigma )z^{\frac{\sigma }{\sigma -1}}+\sigma z}.  \label{eq:z}
\end{equation}%
The above ratio is strictly quasi-concave, so $z^{\ast }(\sigma )$ is
uniquely defined. Hence, the candidate mechanism is a constant surplus
sharing rule. To make the notation more compact, it is useful to provide
notation for the value attained by \eqref{eq:z}: 
\begin{equation}
B(\sigma )=\frac{1-z^{\ast }(\sigma )}{(1-\sigma )(z^{\ast }(\sigma ))^{%
\frac{\sigma }{\sigma -1}}+\sigma z^{\ast }(\sigma )}.  \label{eq:b}
\end{equation}%
We now show that $B(\sigma )$ is indeed an upper bound on the competitive
ratio and that this bound is tight.

\begin{theorem}[Buyer Surplus Guarantee]
\label{thm:osg}$\quad $\newline
For every mechanism $t$, 
\begin{equation}
\inf_{c\in \mathcal{C}_{cx}}\frac{U(c,t)}{W(c)}\leq B(\sigma ).
\label{eq:rat}
\end{equation}%
Furthermore, the inequality is attained as an equality if and only if 
\begin{equation*}
t(q)=z^{\ast }(\sigma )u(q).
\end{equation*}
\end{theorem}

The result provides a mechanism that can guarantee a fraction of the
efficient social surplus regardless of the cost function. Furthermore, the
mechanism is optimal in the sense that it maximizes the share of the
efficient total surplus across all possible mechanisms. The transfer rule is
simple in the sense that it consists of sharing a constant fraction of the
utility $u\left( q\right) $.

To prove Theorem \ref{thm:osg}, we first prove that, if $t(q)\neq z^{\ast
}(\sigma )u(q)$ then we can find $(\widehat{\gamma },\widehat{q})$ such that
the ratio evaluated at $\widehat c$ (see \eqref{eq:plc}) is strictly smaller
than $B(\sigma)$. We then show that, if $t=z^{\ast }u(q)$ then the ratio is
always (weakly) larger than $B(\sigma )$.

We are then left with providing intuition for the fraction of utility that
the buyer shares with the seller, that is, $z^{\ast }(\sigma )$. We plot the
behavior of the surplus sharing rule as a function of the exponent $\sigma $%
\ below in Figure \ref{fig:z}. As an approximation, we can see that $z^{\ast
}(\sigma )\approx \sigma $, and in fact, 
\begin{equation*}
|z^{\ast }(\sigma )-\sigma |\leq 0.12.
\end{equation*}%
Hence, the optimal rule prescribes that the buyer should share with the
seller a fraction that is almost the same as the exponent of the utility
function. For intuition, suppose first that the exponent $\sigma \approx 0$.
Then, an $\varepsilon $ supply of the product can guarantee a utility of
(almost) 1 to the buyer, which is essentially the maximum utility that the
buyer can attain. Hence, the buyer can offer a small transfer as long as
there is some positive supply, which is attained by $t(q)\approx \sigma u(q)$%
. Note that $u^{\prime }(0)=\infty $, so the buyer will always obtain some $%
\varepsilon $ amount of the product. By contrast, suppose that the exponent
is $\sigma \approx 1$\thinspace\ and thus the utility is near linear $u(q)=q$%
. Suppose the buyer offers $t(q)=mq$, for some $m<1$. Then, the buyer
surplus would be zero if the cost is $c(q)=(1+m)q/2$, but the social surplus
would still be positive (in fact, infinite if $u(q)$ is exactly linear).
Hence, we can see that to guarantee a positive ratio as $\sigma \rightarrow
1 $, we must have that $t(q)$ converges to $u(q)$.

We plot the competitive ratio guarantee in Figure \ref{fig:g}. We can see
that $B(\sigma )$ is decreasing in $\sigma $. Furthermore, a linearly
decreasing function in the exponent $\sigma $\ seems like a good
approximation: 
\begin{equation*}
B(\sigma )\approx 1-\sigma .
\end{equation*}%
More precisely, one can check numerically that: 
\begin{equation*}
|B(\sigma )-(1-\sigma )|\leq 0.16.
\end{equation*}%
Hence, as an order of magnitude it is useful to think that the competitive
ratio guarantee is almost linear in $\sigma $. In contrast to Proposition %
\ref{prop:ocss}, we now have that the competitive ratio converges to 0 as $%
\sigma $ approaches 1. We provide an intuition after we provide Theorem \ref%
{thm:rev} because it will be easier to explain this difference in a more
general environment.

\begin{figure}[p]
	\centering
	\includegraphics[width=4.8661in,height=2.4856in,keepaspectratio]{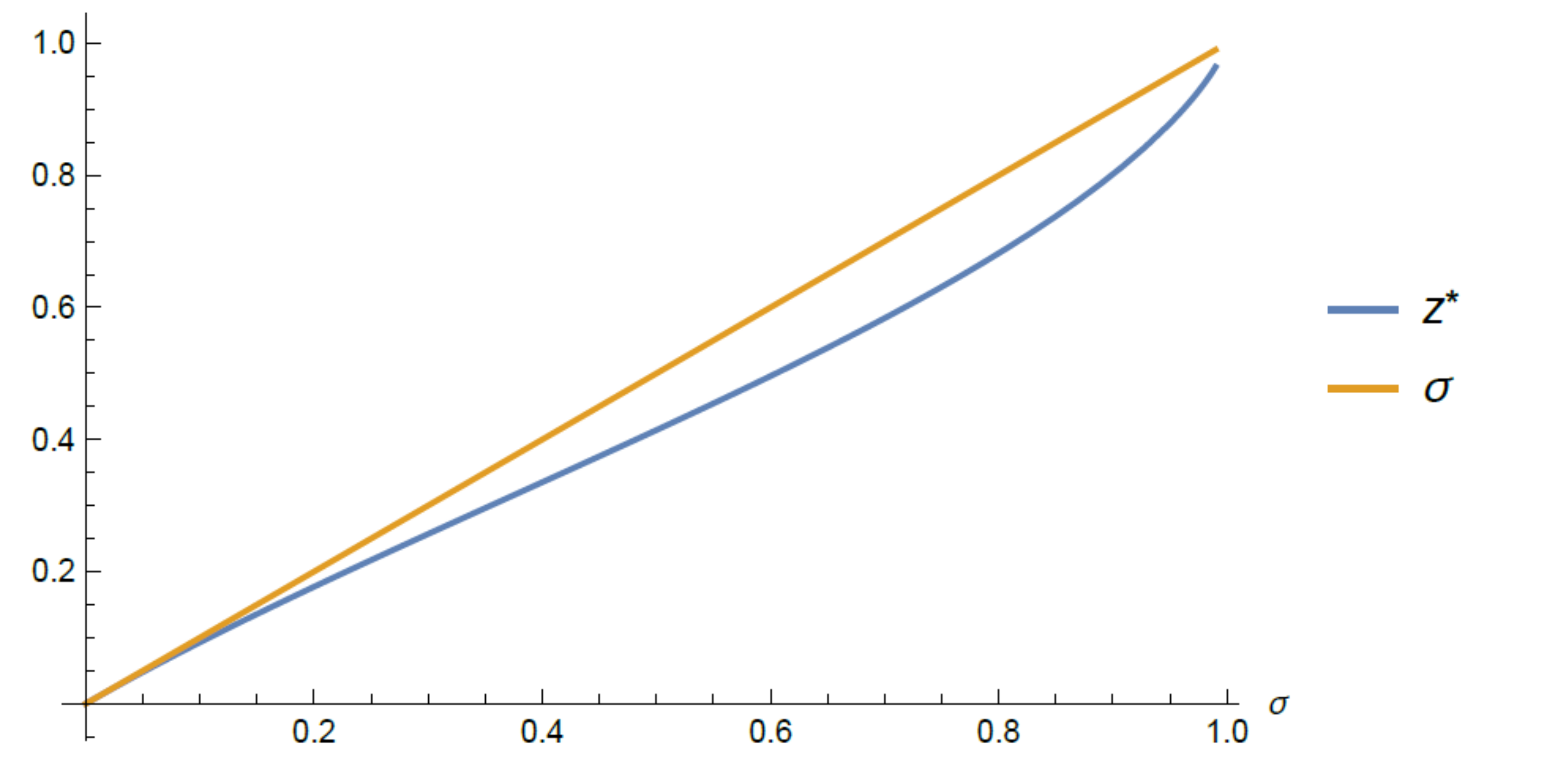}
	\caption{Surplus share $z^{\ast}(\sigma)$ as function of the elasticity $\sigma$.}
	\label{fig:z}
	
	\vspace*{5.3em} 
	
	\includegraphics[width=4.8077in,height=2.7796in,keepaspectratio]{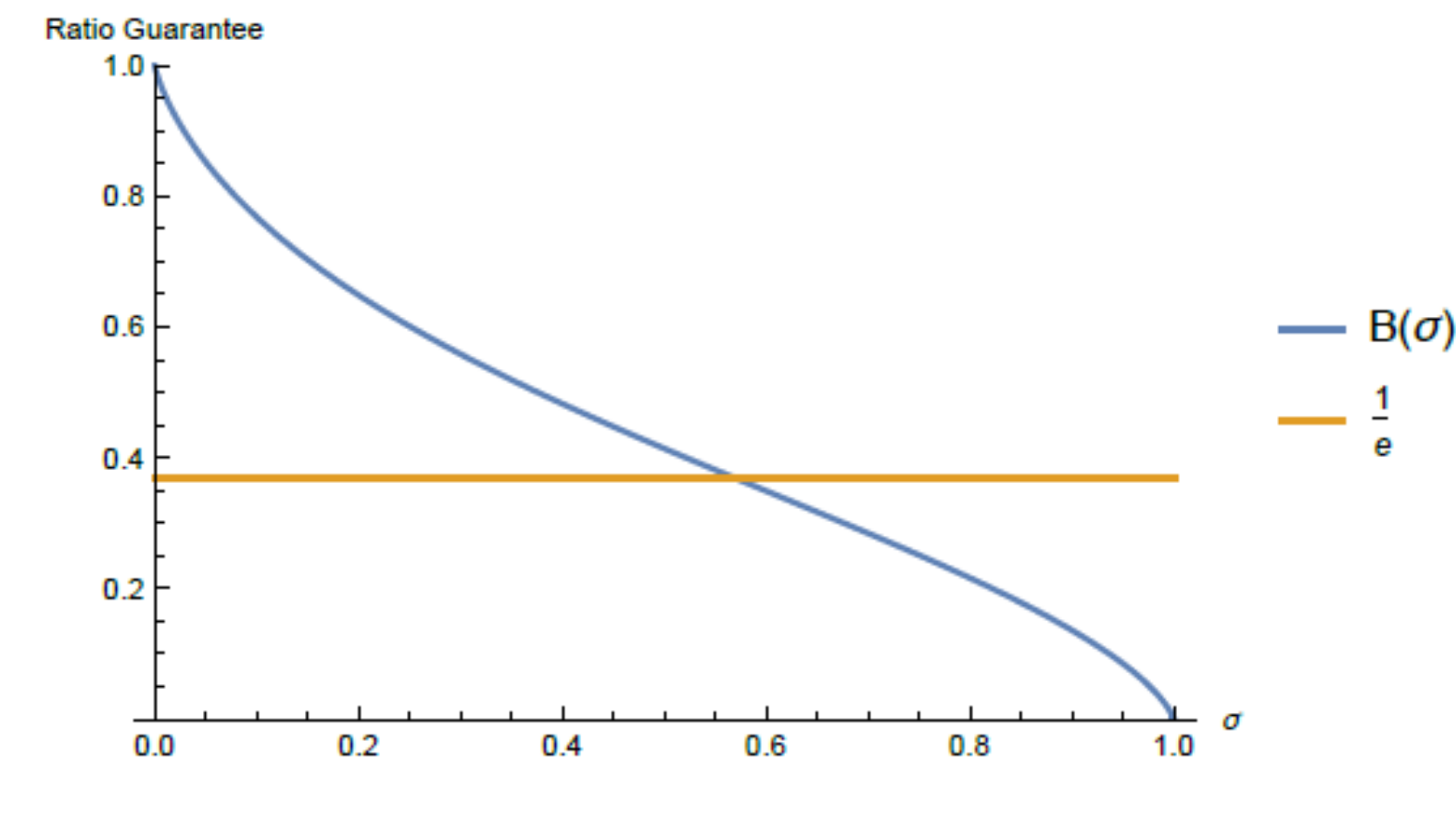}
	\caption{Surplus Share Guarantee as a Function of Elasticity}
	\label{fig:g}
\end{figure}

Theorem \ref{thm:osg} provides the competitive ratio, which is attained by a
constant share mechanism. This simple policy generated a positive ratio
guarantee even when the buyer was restricted to use a deterministic
mechanism. In the Appendix (Section \ref{appB}) we provide an alternative
version of this result in which we allow the buyer to randomize over all
possible mechanisms. We find in Proposition \ref{prop:saddle} that the
mechanism that attains the competitive ratio is a randomization over
constant share mechanisms. Furthermore, we characterize the saddle point,
analogously to Proposition \ref{bmech}. Somewhat surprisingly, the
competitive ratio when the buyer is allowed to randomize does not differ
qualitatively from the one characterized in Theorem \ref{thm:osg}.

\subsection{General Cost and Utility Functions\label{subsec:ext}}

Theorem \ref{thm:osg} used the assumption that the utility of the quantity
is given by a power function and the cost function is weakly convex (in
addition to being increasing). The constant elasticity of demand and the
weak convexity have a natural economic interpretation if the units of $q$
are meaningful, e.g., they represent quantity. But if there are no natural
units, there may not be a natural economic interpretation, e.g., if $q$
represents quality. We now relax the assumptions on utility and cost and
replace them with an assumption that is unit-free. The basic idea of this
section is to obtain a result in this more general environment by reducing
it to arguments we have established in Theorem \ref{thm:osg}.

In this section, we let $u$ be any increasing twice differentiable function,
and let $\mathcal{C}$ consist of all almost everywhere differentiable
increasing functions. \ 

For any given utility function $u\left( q\right) $ of the buyer, we suggest
a change of variable that leads to a new utility function that can be stated
as a power utility. The change of variable induced by the utility function
then generates cost functions under the new variable for every cost function 
$c\in \mathcal{C}$ in the feasible set $\mathcal{C}$. We shall derive
conditions, essentially conditions on demand and supply elasticity, under
which the cost functions under the new variable preserve monotonicity and
weak convexity. This allows us to apply the arguments provided in Theorem %
\ref{thm:osg} much beyond the environment that we consider there.

In more detail, we consider an increasing utility function $u\left( q\right) 
$ and define a new variable based on it:%
\begin{equation}
x\triangleq \left( \sigma u\left( q\right) \right) ^{\frac{1}{\sigma }}
\label{eq:x}
\end{equation}%
for some $\sigma \in \left( 0,1\right) $. From the definition, we obtain
that $q$ as a function of the newly defined variable $x$: 
\begin{equation*}
q=u^{-1}\left( \frac{x^{\sigma }}{\sigma }\right)\text{.}
\end{equation*}%
Hence, the utility function in terms of this newly defined variable can be
written as a power utility function: 
\begin{equation}
\widehat{u}\left( x\right) \triangleq \frac{x^{\sigma }}{\sigma }.
\label{eq:ux}
\end{equation}%
We can then describe any utility function $u$ as a power function under the
right change of variables.

The bound in Theorem \ref{thm:osg} depended on the exponent $\sigma $\ of
the utility function. However, in the change of variables suggested by %
\eqref{eq:ux} we have not yet imposed any discipline on the possible values
of the exponent $\sigma $. However, we cannot apply Theorem \ref{thm:osg}
unless (after the change of variables) every cost function $c\in \mathcal{C}$
is convex. We thus have to find which values of $\sigma $ in \eqref{eq:ux}
are consistent with a convex cost function. More precisely, we consider the
cost function: 
\begin{equation}
\widehat{c}\left( x\right) \triangleq c\left( u^{-1}\left( \frac{x^{\sigma }%
}{\sigma }\right) \right) ,  \label{ctran}
\end{equation}%
which is the cost function written in terms of the variable $x$ instead of $%
q $. The question then arises as to whether under the transformation
suggested by $x$, any newly defined cost function $\widehat{c}\left(
x\right) $ based on $c\left( q\right) $\ is convex in $x$. We provide a
condition based on the curvature of the underlying functions $u\left(
q\right) $ and $c\left( q\right) $ under which $\widehat{c}\left( x\right) $
preserves monotonicity and convexity.

For a given cost and utility function, $c\left( q\right) $ and $u\left(
q\right) $ we consider a joint index of the curvature of the cost and
utility function defined by: 
\begin{equation}
\delta (q,c)\triangleq \frac{u(q)}{qu^{\prime }(q)}\left( \frac{c^{\prime
\prime }\left( q\right) q}{c^{\prime }\left( q\right) }-\frac{qu^{\prime
\prime }(q)}{u^{\prime }(q)}\right) .  \label{ccd00}
\end{equation}%
A desirable property of this measure of curvature is that it is a unit-free
measure. That is, for any increasing function, say $q=h(x)$, if we measure
quality by units of $x$ rather than $q$ the measure $\delta $ does not
change. More precisely, we define: 
\begin{equation*}
\widetilde{u}(x)\triangleq u(h(x))\text{ and }\widetilde{c}(x)\triangleq
c(h(x)),
\end{equation*}%
and $\widetilde{\delta }$ appropriately defined in terms of $x$ (using $%
\widetilde{u}$). It is simple to verify that: 
\begin{equation*}
\widetilde{\delta }(x,\widetilde{c})=\delta (q,{c}).
\end{equation*}%
Hence, non-linear transformations of the units in which $q$ is measured do
not change the measure of curvature (it is easy to verify this is not true
for other operations, like derivatives or elasticities). We can then write $%
\delta $ more conveniently as follows: 
\begin{equation}
\delta (q,c)=\frac{\overline{u}\frac{d^{2}c(u^{-1}(\overline{u}))}{d%
\overline{u}^{2}}}{\frac{dc(u^{-1}(\overline{u}))}{d\overline{u}}}
\label{costfcS}
\end{equation}%
Hence, $\delta $ provides the curvature of the cost-- measure by the
elasticity of marginal cost-- when measured in utils. In fact, this is also
a measure of the curvature of the social welfare when measured in utils
because, by construction, the utility function is linear when measuring
quality in utils. We denote by $\underline{\delta }$ the lowest value
attained by $\delta (q,c)$: 
\begin{equation*}
\underline{\delta }\triangleq \inf_{(q,c)\in \mathbb{R}_{+}\times \mathcal{C}%
}\delta (q,c)\text{ with }(q,c)\text{ such that }q\leq \overline{q}(c).
\end{equation*}%
In other words, we find the infimum over quantity and costs, considering
only quantities below the efficient level. We assume that: 
\begin{equation}
\underline{\delta }>0.  \label{assud}
\end{equation}%
This implies that the law of diminishing returns applies to our setup when
cost is measured in utils (that is, \eqref{costfcS} is positive). We can now
discipline the choice of the exponent $\sigma $ in the change of variable
given by \eqref{eq:ux}.

\begin{lemma}[Convexity of the Transformed Cost]
\label{lem:ctc}\quad \newline
The cost function $\widehat{c}\left( x\right) $ (see\eqref{ctran}) is convex
for every $c\in \mathcal{C}$ and at every $x\in \mathbb{R}_{+}$ if and only
if 
\begin{equation}
\sigma \geq \frac{1}{1+\underline{\delta }}.  \label{dcds}
\end{equation}
\end{lemma}

We can now apply the results in Theorem \ref{thm:osg} by performing the
change of variables in \eqref{eq:x} subject to $\sigma $ satisfying %
\eqref{dcds}, and in fact, we will get the sharpest bound when \eqref{dcds}
is satisfied with equality. We thus define: 
\begin{equation}
\widehat{\sigma }\triangleq \frac{1}{1+\underline{\delta }}  \label{eq:sig}
\end{equation}%
Thus, a strongly positive index allows for a lower exponent $\widehat{\sigma 
}$ and in turn for an improved bound $B\left( \widehat{\sigma }\right) $.

\begin{theorem}[Competitive Ratio with General Cost and Utility Functions]
\qquad \label{thm:rev}\quad \newline
The mechanism $t(q)=z^{\ast }(\widehat{\sigma })u(q)$ attains at least
ratio: 
\begin{equation*}
\min_{c\in \mathcal{C}}\frac{U(c,t)}{W(c)}\geq B(\widehat{\sigma })
\end{equation*}
\end{theorem}

This theorem generalizes Theorem \ref{thm:osg} by showing that a constant
surplus sharing rule provides a positive ratio guarantee in general utility
and cost functions environments. In particular, our simple mechanism can
guarantee a share $B(\widehat{\sigma })$ of the efficient social surplus
(with $\widehat{\sigma }$ appropriately defined in terms of $\underline{%
\delta }$). We note that \eqref{assud} guarantees that $\widehat{\sigma }\in
(0,1)$, which in turn guarantees that the ratio guarantee is indeed strictly
positive. When \eqref{assud} is not satisfied, $\widehat{\sigma }$ attains
values greater than 1. In this case, our results in Section \ref%
{subsec:simple} do not apply, and consequently, we are not able to guarantee
of positive ratio.

Intuitively, the buyer can guarantee a positive ratio only if social surplus
is sufficiently concave in $q$ (measured by the concavity of \eqref{ctran}).
In contrast to Theorem \ref{thm:osg}, we now only show that $t=z^{\ast }(%
\widehat{\sigma })u(q)$ attains $B(\widehat{\sigma })$ but we do not show
that this is optimal (that is, that any other mechanism attains less). The
reason is that we do not necessarily have that for every convex function $%
\widehat{c}\in \mathcal{C}_{cx}$, there exists $c\in \mathcal{C}$ such that %
\eqref{ctran} is satisfied. In other words, the set of feasible cost
functions $\mathcal{C}$ might be too small to guarantee that any other
mechanism performs worse than $t=z^{\ast }(\widehat{\sigma })u(q)$. It is
easy to verify that if the set: 
\begin{equation*}
\widehat{C}\triangleq \{\widehat{c}\in C_{cx}:\text{ there exists }c\in 
\mathcal{C}\text{ such that }\widehat{c}\text{ if given by \eqref{ctran}.}\}
\end{equation*}%
satisfies that $\widehat{C}=C_{cx}$ (note that by definition $\widehat{C}%
\subseteq C_{cx}$) then $t=z^{\ast }(\widehat{\sigma })u(q)$ indeed attains
the competitive ratio.

\section{Related Applications:\ Nonlinear Pricing and Regulation\label%
{sec:coba}}

We now study two problems closely related to the procurement problem. First,
we examine the optimal mechanism for a regulator who maximizes a linear
combination of buyer surplus and profits. Hence, unlike the procurement
setting, the regulator also places positive weight on profits. Second, we
examine the pricing rule for a seller facing unknown demand, the nonlinear
pricing problem as defined by \cite{muro78}. In both cases, we find that a
simple rule can guarantee a positive fraction of the efficient social
surplus; furthermore, the simple mechanism is optimal in the sense of the
competitive ratio.

\subsection{Regulation}

We now examine situations in which there is a regulator who maximizes a
linear combination of buyer surplus and profits. The buyer surplus $U$ and
seller surplus $\Pi $ are determined by the allocation $q\left( c,t\right) $
and the mechanism $t\left( q\left( c,t\right) \right) $\ for the output, as
given by \eqref{cs0} and \eqref{eq:ps}. We assume the regulator places a
weight of 1 on buyer surplus and a weight $\alpha \ $on profits:%
\begin{equation*}
\alpha \in \left[ 0,1\right] .
\end{equation*}

In the absence of a common prior over the class $\mathcal{C}$\ of
permissible cost functions, the regulator cannot compute the \emph{expected}
regulator surplus of any given mechanism across all possible cost functions.
Therefore, we consider the ratio between the regulator surplus under
incomplete information and the efficient social surplus under complete
information. Formally, the regulator is choosing a mechanism $t$ while
nature is choosing a cost function $c$:$\ $%
\begin{equation*}
\sup_{t}\inf_{c}\frac{U\left( c,t\right) +\alpha \Pi \left( c,t\right) }{%
W\left( c\right) }.
\end{equation*}%
The social welfare $W\left( c\right) $ depends on the realized cost function 
$c$ but not on the mechanism. The regulator seeks to identify the mechanism
that attains the largest ratio of weighted regulator surplus against the
efficient social surplus across all possible cost functions $c\in \mathcal{C}
$.

To provide our next result, we go back to our baseline model in Section \ref%
{subsec:simple}. That is, we assume that $u(q)=q^{\sigma }/\sigma $ and $%
\mathcal{C}$ is the class of all increasing and convex functions. As before,
we start with a constant share mechanism: 
\begin{equation*}
t(q)=z_{\alpha }^{\ast }u(q),
\end{equation*}%
where now $z_{\alpha }^{\ast }$ solves: 
\begin{equation}
z_{\alpha }^{\ast }\triangleq \arg \max_{z\in \lbrack 0,1]}\frac{1-(1-\alpha
)z}{(1-\sigma )z^{\frac{\sigma }{\sigma -1}}+\sigma z}.  \label{eq:z2}
\end{equation}%
The objective function is strictly quasi-concave, so $z_{\alpha }^{\ast }$
is uniquely defined. Hence, the price schedule is a constant buyer surplus
share mechanism. As before, to make the notation more compact, it is useful
to provide notation for the value attained by \eqref{eq:z2}: 
\begin{equation*}
B_{\alpha }(\sigma )\triangleq \frac{1-(1-\alpha )z_{\alpha }^{\ast }}{%
(1-\sigma )(z_{\alpha }^{\ast })^{\frac{\sigma }{\sigma -1}}+\sigma
z_{\alpha }^{\ast }}
\end{equation*}%
We now show that $B_{\alpha }(\sigma )$ is the upper bound on the ratio and
that this bound is tight, thus attaining the competitive ratio.

\begin{theorem}[Competitive Ratio of Regulation]
\label{thm:wel}$\quad $\newline
For every mechanism $t(q)$, 
\begin{equation*}
\min_{c(q)}\frac{U\left( c,t\right) +\alpha \Pi \left( c,t\right) }{W\left(
c\right) }\leq B_{\alpha }(\sigma )
\end{equation*}%
Furthermore, the inequality is tight if and only if $t(q)=z_{\alpha }^{\ast
}u(q)$.
\end{theorem}

Theorem \ref{thm:wel} establishes that Theorem \ref{thm:osg} extends to
situations in which the buyer is a regulator who seeks to maximize a linear
combination of buyer surplus and profits. The current model of regulation
follows closely the work of \cite{bamy82} who ask how to regulate a
monopolist with private information about their cost. There, the regulator,
as the public agency, seeks to maximize the buyer surplus or a weighted sum
of buyer surplus and monopoly profit. As the weight $\alpha \in \left[ 0,1%
\right] $ that the regulator assigns to the profit of the firm increases,
the objective of the regulator becomes closer to the benchmark of the social
welfare. A natural comparative static result then emerges.

\begin{corollary}[Comparative Statics of Competitive Ratio]
$\quad $\newline
As the weight on profit increases, the competitive ratio $B_{\alpha }\left(
\sigma \right) $ increases and converges to $1$ as $\alpha$ converges to 1
for all $\sigma \in (1,\infty ).$
\end{corollary}

\subsection{Nonlinear Pricing}

We now consider the nonlinear pricing problem. Here, the (representative)
buyer has private information about their willingness to pay and the seller
has to offer a tariff in the presence of uncertainty about the willingness
to pay. Thus, we invert the roles of buyer and seller. Namely there is a
seller with a known cost function: 
\begin{equation*}
c(q)=\frac{\sigma -1}{\sigma }\frac{q^{\sigma }}{\sigma },
\end{equation*}%
with $\sigma >1$. The buyer has a utility function: 
\begin{equation*}
u(q)=vq,
\end{equation*}%
for some value (willingess-to-pay) $v\in \mathbb{R}_{+}$. The value $v$ is
known to the buyer but unknown to the seller. The seller offers a transfer $%
t(q)$ to the buyer, who then chooses a quantity that maximizes their net
utility: 
\begin{equation*}
q(v,t)=\max \{u(q)-t(q)\}.
\end{equation*}%
The profits are given by: 
\begin{equation*}
\Pi (v,t)=t(q(v,t))-c(q(v,t)).
\end{equation*}%
Following the same notation as before, we denote by $W(v)$ the efficient
social surplus when the buyer's value is $v$. We thus obtain the same model
as \cite{muro78}, with the additional constraint that the cost function is a
power function. However, instead of characterizing the Bayes optimal
mechanism, we provide a mechanism that attains a positive ratio guarantee.
Furthermore, it attains the competitive ratio.

\begin{proposition}[Competitive Ratio for Nonlinear Pricing]
\label{prop:seller}\quad \newline
For every mechanism $t(q)$, 
\begin{equation*}
\min_{v\in \mathbb{R}_{+}}\frac{\Pi (v,t)}{W(v)}\leq \sigma ^{\frac{\sigma }{%
1-\sigma }}.
\end{equation*}%
Furthermore, the inequality is tight if and only if $t(q)=\sigma c(q)$.
\end{proposition}

We thus obtain the same result as in Proposition \ref{prop:ocss}. While the
expression for the competitive ratio in terms of $\sigma $ remains the same,
we now have that $\sigma >1$, so we now obtain a lower competitive ratio
relative to the procurement problem. But similar to the procurement problem
with linear cost, the competitive ratio in the nonlinear pricing environment
satisfies the saddle-point property and can be interpreted as the solution
to a Bayesian optimal nonlinear pricing problem. Thus, an analogue to
Proposition \ref{bmech} exists and can be found in an earlier working paper
version (\cite{behm23a}).

As prescribed by Proposition \ref{prop:ocss}, we now have that the optimal
pricing is: 
\begin{equation*}
t(q)=\sigma c(q).
\end{equation*}%
We now have that $\sigma >1$, so as expected, the transfer is always greater
than the cost. Finally, note that by construction: 
\begin{equation*}
\frac{t(q)-c(q)}{t(q)}=\frac{\sigma-1 }{\sigma }.
\end{equation*}%
Hence, the optimal pricing rule is a constant markup rule where the
principal always demands a fixed markup over cost. The common use of a
constant markup rule was also already observed by \cite{hahi39} who document
from interviews, the use of \textquotedblleft full-cost\textquotedblright\
or \textquotedblleft cost-plus\textquotedblright\ pricing.~ \cite{sche90}
notes that \textquotedblleft Hall and Hitch and later analysts found several
reasons why businessmen use cost-based rules of thumb in their pricing
decisions\ldots .\textquotedblright .~ One is that \textquotedblleft it is a
way of coping with (essentially by ignoring) uncertainties in the estimation
of demand function shapes and elasticities\textquotedblright .

\begin{figure}[t]
	\centering
	\includegraphics[width=5.2592in,height=3.0134in,keepaspectratio]{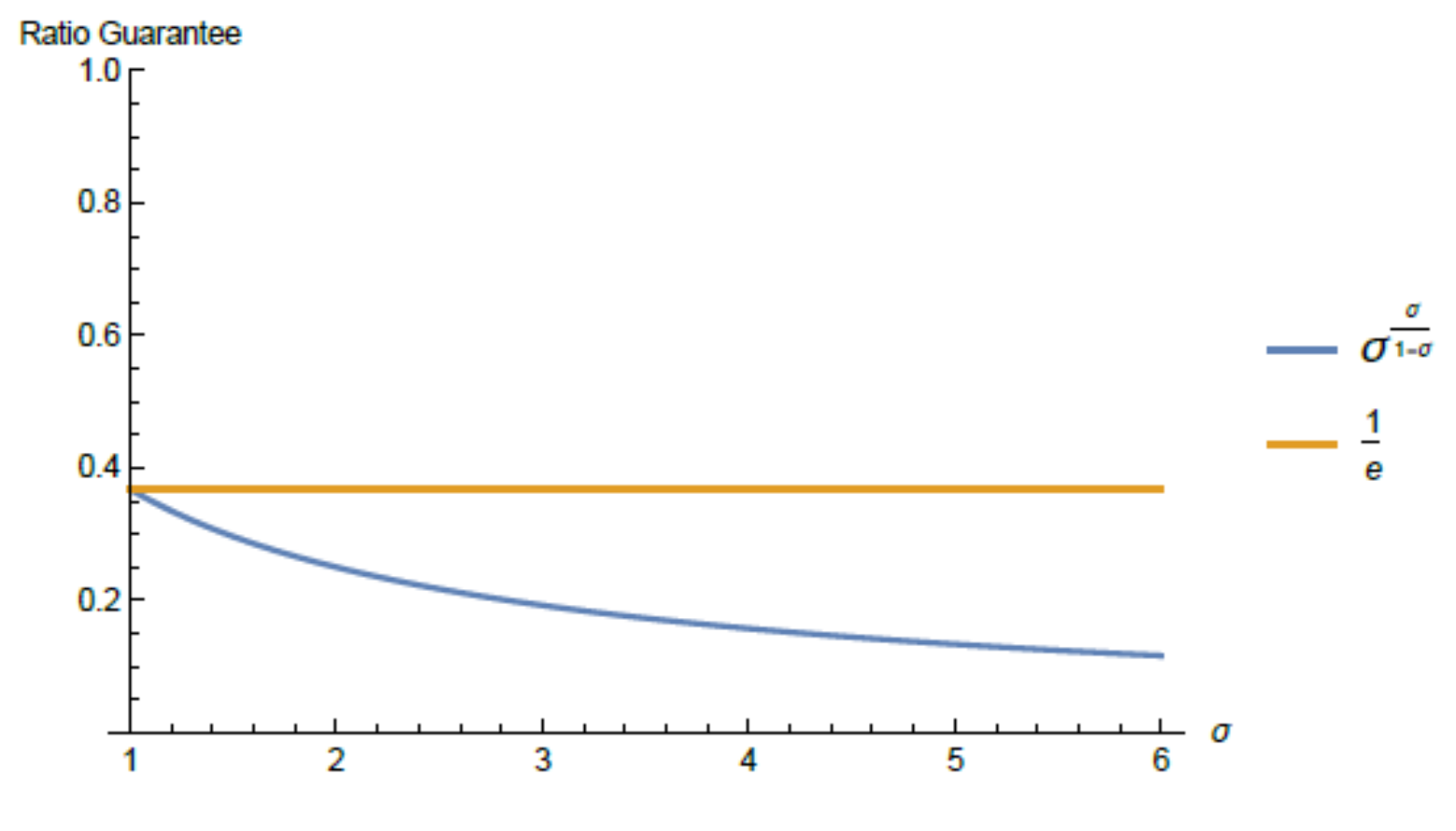}
	\label{fig:sm}
\end{figure}

\section{Alternative Criteria\label{sec:maxmin}}

We proposed a class of simple mechanisms-- constant utility share
mechanisms-- for environments where the buyer has minimal information about
the seller. The conventional approach of Bayesian utility maximization was
therefore not available. We suggested that the competitive ratio may be a
suitable criterion to evaluate the performance of a simple rule. In this
final section, we briefly discuss different but conceptually related
criteria used in the absence of a Bayesian prior distribution. We conclude
by analyzing a variation of the competitive ratio itself, where we replace
the benchmark of the social surplus by the Bayes optimal utility. We show
that these variations lead to qualitatively as well as quantitatively very
similar results in support of constant share mechanisms.

\subsection{Alternative Robustness Criteria\label{subsec:maxmin}}

The \emph{maximin ratio} is a commonly used criterion in the absence of a
Bayesian prior distribution. There are related criteria such as the \emph{%
maximin utility} or \emph{maximin regret.} \cite{anbb25b} provide a unifying
framework that they refer to as $\lambda $ regret and show that all of the
above robust decision criteria can be obtained as special cases by choosing
the weight $\lambda $ in an additive objective: 
\begin{equation*}
\max_{t\in \mathcal{T}}\min_{c\in \mathcal{C}}\left\{ U\left( c,t\right)
-\lambda W\left( c\right) \right\} ,
\end{equation*}%
with $\lambda \in \left[ 0,1\right] $. The maximin regret is obtained with $%
\lambda =1$, the maximin utility with $\lambda =0$, and the maximin ratio
solution is obtained by the value of $\lambda \in \left( 0,1\right) $ at
which the above $\max \min $ problem attains the value of zero. We now
briefly discuss the advantages of the competitive ratio in the procurement
setting.

For concreteness we consider our baseline payoff environment introduced in
Section \ref{sec:model} and change the criteria used to evaluate different
mechanisms. That is, we consider a buyer who is procuring $q$ units of a
product from a seller. The buyer's utility is: 
\begin{equation*}
u(q)=\frac{q^{\sigma }}{\sigma },
\end{equation*}%
and the seller has a cost function $c$ that is increasing in the quantity $q$%
. Let $\mathcal{C}$ be a set of feasible cost functions-- possibly different
than the ones considered thus far-- and consider the following maximin
utility problem: 
\begin{equation}
\max_{t\in \mathcal{T}}\ \min_{c\in \mathcal{C}}U(c,t),  \label{min}
\end{equation}%
If $\mathcal{C}$ is the set of all convex cost functions, $\mathcal{C=C}%
_{cx},$ then we clearly will have that: 
\begin{equation*}
\inf_{c(q)\in \mathcal{C}_{cx}}U(c,t)=0
\end{equation*}%
for all mechanisms $t$. After all, the adversarial nature can choose a cost
function sufficiently large so that the socially efficient allocation is
arbitrarily close to zero. In this case, even the best mechanism cannot
generate any positive net buyer utility. Moreover, all payment functions
that are below the gross buyer surplus, i.e., $t\left( q\right) \leq U\left(
q\right) $\ are \textquotedblleft optimal\textquotedblright . Thus, the
prediction arising from the maximin utility approach is very weak. Suppose
now that the set of feasible cost functions $\mathcal{C}$ is compact and
that they generate a minimum positive social welfare $\underline{W}$: 
\begin{equation*}
\underline{W}=\min_{c(q)\in \mathcal{C}}W(c).
\end{equation*}%
That is, $\underline{W}$ is the minimum social surplus across all cost
functions. A mechanism that solves \eqref{min} is: 
\begin{equation}
t(q)=u(q)-\underline{W}.  \label{ccd2}
\end{equation}%
In other words, the buyer gives the efficient social surplus to the seller,
except for the fixed component $\underline{W}.$

To check that \eqref{ccd2} indeed solves \eqref{min} it suffices to check
that the buyer cannot guarantee himself more than $\underline{W}$ and that
the above mechanism indeed guarantees $\underline{W}$. Since there is a
non-trivial fixed price, it is necessary to satisfy the participation
constraint. However, by construction, the seller will solve: 
\begin{equation*}
\max_{q}\left\{ t(q)-c(q)\right\} =W(c)-\underline{W}\geq 0,
\end{equation*}%
where the inequality follows from the definition of $\underline{W}.$

The discussion illustrates that there are two important drawbacks to using
maximin mechanisms that are not present when using the competitive ratio.

First, from a normative perspective, the optimal maximin rule is determined
by the lowest type (in this case, the highest possible cost in $\mathcal{C}$%
). Hence, it is very sensitive to the assumptions about the family of cost
structures $\mathcal{C}$ that the buyer must consider possible. If the buyer
makes a mistake in specifying $\mathcal{C}$ and does not consider a cost
function that allows generating $\varepsilon $ less surplus, then the
optimal rule will guarantee zero buyer surplus. Hence, the buyer is left
with a complicated problem of finding a class of cost functions that is
sufficiently large to prevent any misspecification, but also restrictive
enough that allows to guarantee of some buyer surplus. \cite{baka24} study
the discontinuous impact of the domain of uncertainty in maximin problems,
and show that this is a more general phenomenon when using a maximin
approach.

Second, and from a positive perspective, there is a plethora of optimal
maximin mechanisms. Hence, it is difficult to provide sharp predictions
about the optimal policy, and simple mechanisms do not naturally arise as a
unique optimal solution. An alternative to selecting an optimal maximin rule
is by considering the one that performs the best against some parametrized
Bayesian class of cost functions. In this case, the rule will perform well
against the worst-case scenario, and can perform the best possible against
some parametrized class of cost functions. This is the approach undertaken
by \cite{mipp25}.

While we have argued that the competitive ratio serves as a more instructive
benchmark than maximin utility, we emphasize that the discussion here is
phrased for the specific problem that we analyze. For example, in other
economic environments one does not need to bound or restrict the space of
uncertainty to obtain a well-defined maximin problem. This is usually the
case when the space of uncertainty does not have such a direct impact on the
scale of the objective function. For example, when the uncertainty is about
higher order beliefs (see \cite{brdu21}, \cite{brdu25a}) or about the
correlation in the values across many goods (see \cite{carr17}), then the
realization of the non-Bayesian uncertainty does not change the efficient
social surplus. Hence, in these examples, there is no clear
\textquotedblleft lowest type\textquotedblright\ in the problem. Hence,
there is no unique type that determines the optimal maximin rule. In this
type of problem, maximin objectives also sometimes lead to sharp predictions
about the optimal rule, and frequently to simple mechanisms. Hence, what is
crucial about the competitive ratio that makes it useful in this context is
that it is a scale-free benchmark, which is useful in the context of our
problem.

While we have focused our discussion on optimal maximin mechanisms, the same
kind of concerns arise when using mechanisms that minimize regret (see, for
example, \cite{gush25}). The main difference is that mechanisms that
minimize regret are tailored to perform well against the highest type
(lowest cost), but the main ideas of the discussion go through unchanged.
The key benefit of using the competitive ratio is that it is a scale-free
benchmark, so it can be applied without any restrictions on the set of
problems considered. Hence, the rule we find is naturally simple, as there
are no parameters about the unknown supply that can be used to tailor the
optimal mechanism.

\subsection{Alternative Versions of the Competitive Ratio}

\label{sec:bopt}

So far, we have evaluated the performance of constant share mechanisms by
computing the buyer's surplus relative to the efficient social surplus, and
we have argued that this is a good benchmark in the absence of Bayesian
uncertainty. We now discuss an alternative notion of competitive ratio with
an alternative motivation.

Suppose now the buyer had Bayesian uncertainty about the seller's cost, but
the buyer was constrained to using a constant share mechanism. We now
compare the performance of constant share mechanisms relative to the
Bayes-optimal mechanism. We interpret this constraint as capturing the fact
that computing the optimal transfer when the Bayesian prior is over all
possible convex cost functions may be an intractable problem.

To formalize this problem, for any $G\in \Delta (\mathcal{C}_{cx})$ we
define: 
\begin{equation*}
U^{\ast }(G)\triangleq \sup_{t}\int U(c,t)dG(c).
\end{equation*}%
That is, $U^{\ast }(G)$ is the buyer surplus generated by the Bayes-optimal
mechanism when the distribution over cost is $G$. The performance of the
constant share mechanism relative to the Bayes-optimal mechanism is defined
as follows: 
\begin{equation}
\min_{G\in \Delta (\mathcal{C}_{cx})}\max_{z\in \lbrack 0,1]}\frac{\int
U(c,z)dG(c)}{U^{\ast }(G)}.  \label{fff2}
\end{equation}%
Thus, we compute the performance of the simple rule relative to the
Bayes-optimal rule. We emphasize that in this case, we allow for the
constant share mechanism to be adapted to the distribution over cost. This
is consistent with the current motivation for choosing a constant share
mechanism, which is the impossibility of implementing a Bayesian-optimal
mechanism due to its complexity (instead of the absence of a Bayesian prior).

In the Appendix (Section \ref{appB}) we characterize the value of the
competitive ratio (Proposition \ref{prop:de}). Somewhat surprisingly, we
obtain the same value as the competitive ratio when we consider our original
benchmark-- that is, when the social surplus is in the denominator-- but we
allow for randomized policies (which is also the case when we can construct
a saddle point for the original problem, see Proposition \ref{prop:saddle}).
While the problems are distinct, the value of the competitive ratio is
qualitatively similar to the one obtained in Theorem \ref{thm:osg}.

Some of the literature has also considered the buyer surplus generated by
the optimal rule--as in\eqref{fff2}-- instead of the efficient social
welfare. For example, \cite{roge03} and \cite{chsa07} compare the
performance of simple mechanisms relative to the Bayes optimal mechanisms in
specific parametric environments. This approach is also pursued in \cite%
{cadw19} with notable applications to multi-unit optimal pricing and
multi-unit auctions. These are problems where the solution to the optimal
Bayesian problem is either unknown or computationally complex, and hence the
question arises whether a simple solution adapted to the distribution of
uncertainty can attain a good approximation in the sense of the ratio
guarantee.

`

\section{Conclusion\label{sec:con}}

This paper establishes that simple mechanisms can resolve complex screening
problems. When buyers lack distributional knowledge about supplier costs,
constant surplus sharing based solely on demand elasticity achieves optimal
worst-case ratio guarantees. The elegance of this solution---sharing a
fraction of buyer surplus proportional to the demand elasticity---provides
both theoretical insight and practical guidance.

The optimal sharing constants and competitive ratios depend on a single
parameter: the elasticity of demand (or supply in selling problems). This
parsimony is striking. A procurement officer need not estimate cost
distributions, but merely understand their organization's demand
responsiveness to price. Similarly, regulators need only observe the demand
elasticity.

Our results extend substantially beyond the baseline procurement setting.
The same principles generate optimal mechanisms for sellers facing unknown
demand and regulators balancing buyer and producer welfare. In each setting,
simple rules depending only on known elasticities achieve competitive ratios
that cannot be improved by any mechanism, regardless of complexity.

While our sharpest results assume constant elasticity and convex costs,
Theorem \ref{thm:rev} demonstrates that a positive ratio guarantee obtains
under much weaker conditions---essentially, that social surplus exhibits
diminishing returns when measured in utility units. This "unit-free"
condition accommodates variable elasticity and non-constant returns to scale.

The theoretical foundations we provide help explain the prevalence of
cost-plus contracts and markup pricing in practice. These seemingly naive
mechanisms emerge as optimal responses to distributional uncertainty. By
characterizing performance guarantees as functions of elasticities, we offer
concrete guidance for mechanism design: more elastic environments require
larger concessions but yield weaker guarantees.

\pagebreak

\section{Appendix: Omitted Proofs}

The Appendix collects the proofs. Before we begin with the proof of Theorem %
\ref{thm:osg}, we introduce some notation and provide an auxiliary lemma. We
denote by 
\begin{equation*}
\mathrm{cav}[t]
\end{equation*}%
the concavification of the function $t:\mathbb{R}_{+}\rightarrow \mathbb{R}%
_{+}$, that is, the smallest concave function that is everywhere larger than 
$t$. For any $\widehat{q}$ such that $\mathrm{cav}[t](\widehat{q})>t(%
\widehat{q})$ we denote by $q_{1},q_{2}$ the support of the concavification
at $\widehat{q}\in \lbrack q_{1},q_{2}]$, that is, $\mathrm{cav}[t](q)$ is
linear in $[q_{1},q_{2}]$, $\mathrm{cav}[t](q_{1})=t(q_{1})$, and $\mathrm{%
cav}[t](q_{2})=t(q_{2})$. Finally, we define: 
\begin{equation*}
z(q)\triangleq \frac{t(q)}{u(q)}\text{ and }{Z}(q)\triangleq \frac{\mathrm{%
cav}[t](q)}{u(q)}.
\end{equation*}%
We provide some properties of these two functions, $z\left( q\right) $ and $%
Z\left( q\right) $, which will be later used in the main part of the proof.

\begin{lemma}[Concavification of $t$]
\label{lemm3}\qquad

\begin{enumerate}
\item \label{lemm31}If $\widehat{q}$ is such that $\mathrm{cav}[t]^{\prime }(%
\widehat{q})=c^{\prime }(\widehat{q})$ then 
\begin{equation}
\widehat{q}\in \underset{q\in \mathbb{R}}{\arg \max }\ \left\{ \mathrm{cav}%
[t](q)-c(q)\right\} .  \label{maxcav1}
\end{equation}

\item \label{lemm32}If $\widehat{q}$ is such that $\mathrm{cav}[t](\widehat{q%
})=t(\widehat{q})$ and \eqref{maxcav1} is satisfied, then 
\begin{equation}
\widehat{q}\in \underset{q\in \mathbb{R}}{\arg \max }\ \left\{
t(q)-c(q)\right\} .  \label{maxcav2}
\end{equation}

\item \label{lemm33}If $\widehat{q}$ is such that $\mathrm{cav}[t](\widehat{q%
})>t(\widehat{q})$ and ${Z}^{\prime }(\widehat{q})<0$ then, ${Z}^{\prime
}(q_{1})<0\text{ or }{Z}^{\prime }(q_{2})<0.$%

\item \label{lemm34}If $\widehat{q}$ is such that $\mathrm{cav}[t](\widehat{q%
})>t(\widehat{q})$ and ${Z}^{\prime }(\widehat{q})\geq 0$ then, 
\begin{equation}
\widehat{q}{Z}^{\prime }(\widehat{q})\geq \min \left\{ q_{1}{Z}^{\prime
}(q_{1}),q_{2}{Z}^{\prime }(q_{2})\right\} .  \label{vfvrr}
\end{equation}
\end{enumerate}
\end{lemma}

\begin{proof}
We establish the above properties of $t$ and $\mathrm{cav}[t]$ serially.

(1.) By construction $\mathrm{cav}[t](q)$ is concave, while by assumption $c$
is convex. Hence, the objective function is concave, so the first-order
condition is sufficient for optimality. \ 

(2.) This follows from the fact that by construction $\mathrm{cav}[t](q)\geq
t(q)$ pointwise, so if $\widehat{q}$ solves \eqref{maxcav1} and this attains
the same value when replacing with $t$, then it must also solve %
\eqref{maxcav2}.

(3.) We note that: 
\begin{equation*}
{Z}^{\prime }(q)=\frac{u^{\prime }(q)}{\left( u(q)\right) ^{2}}\left( \frac{q%
}{\sigma }\mathrm{cav}[t]^{\prime }(q)-\mathrm{cav}[t](q)\right) .
\end{equation*}%
We now note that, if $\mathrm{cav}[t](q)>t(q)$, then the term inside the
parenthesis is linear in $q$ for all $q\in \lbrack q_{1},q_{2}]$. Hence, if
the term inside the parentheses is negative at $\widehat{q}$, then it must
also be negative at $q_{1}$ or $q_{2}$. \ 

(4.) We prove that, if $\mathrm{cav}[t](\widehat{q})>t(\widehat{q})$ and ${Z}%
^{\prime }(\widehat{q})\geq 0$, then $q{Z}^{\prime }(q)$ is concave in $q\in
\lbrack q_{1},q_{2}]$. If $\mathrm{cav}[t](\widehat{q})>t(\widehat{q})$,
then $\mathrm{cav}[t](q)$ is linear for all $q\in \lbrack q_{1},q_{2}]$, so
for some $\alpha ,\beta \in \mathbb{R}$: 
\begin{equation*}
\mathrm{cav}[t](q)=\alpha +\beta q.
\end{equation*}%
We note that $\mathrm{cav}[t]$ is increasing, concave and $\mathrm{cav}%
[t](0)=0$. Hence, we must have that $\alpha,\beta\geq0.$ We now note that: 
\begin{equation*}
\frac{d^{2}}{dq^{2}}\left( q{Z}^{\prime }(q)\right) =-\sigma ^{2}q^{-\sigma
-2}\left( \alpha \sigma (\sigma +1)+\beta q(1-\sigma )^{2}\right) \leq 0.
\end{equation*}%
We thus obtain that $q{Z}^{\prime }(q)$ is concave in $q\in \lbrack
q_{1},q_{2}]$, which implies that \eqref{vfvrr} is satisfied.
\end{proof}

\begin{proof}[Proof of Theorem \protect\ref{thm:osg}]
First, we prove that the inequality (\ref{eq:rat}) is satisfied for every
mechanism $t\left( q\right) $\ and that the inequality (\ref{eq:rat}) is
strict whenever $t(q)\neq z^{\ast }(\sigma )u(q)$. Second, we establish that
the inequality (\ref{eq:rat}) is attained as an equality with the transfer
rule $t(q)=z^{\ast }u(q)$. Thus, we establish that the bound (\ref{eq:b}) is
an upper bound for the competitive ratio, and then that the upper bound (\ref%
{eq:b}) \ can be attained by a constant share mechanism (\ref{eq:ss}). To
make the notation more compact, we omit the argument of $z^{\ast }(\sigma )$%
, and write simply $z^{\ast }$.

\textbf{(Part I: Upper Bound) }It is without loss of generality to assume
that $t\left( q\right) $ is monotone increasing (as the seller will never
choose a dominated quality). We assume without loss of generality that $t(q)$
is differentiable; we can approximate any monotonic function by a
differentiable function pointwise arbitrarily well, and the ratio will
converge appropriately. We fix some $\widetilde{q}\in \mathbb{R}_{+}$ that
we will specify later and define: 
\begin{equation*}
\widetilde{z}\triangleq \frac{t(\widetilde{q})}{u(\widetilde{q})}\text{ and }%
\widetilde{\gamma }\triangleq t^{\prime }(\widetilde{q}).
\end{equation*}%
Hence, the tilde indicates that we are evaluating at $\widetilde{q}$.
Following the definition of $z(q)$ we can write $\widetilde{\gamma }$ as
follows: 
\begin{equation}
\widetilde{\gamma }=\widetilde{z}u^{\prime }(\tilde{q})+z^{\prime }(\tilde{q}%
)u(\tilde{q}).  \label{ccsa}
\end{equation}%
We show that the inequality (\ref{eq:rat}) is satisfied when the cost
function is a piecewise linear cost function of the form: 
\begin{equation}
c(q)=%
\begin{cases}
0, & \text{if }q<\widetilde{q}; \\ 
\widetilde{\gamma }(q-\widetilde{q}), & \text{if }q\geq \widetilde{q}.%
\end{cases}
\label{eq:cost}
\end{equation}%
where (as before) $\widetilde{q}$ is specified later. By construction $%
\widetilde{q}$ satisfies the first-order condition: 
\begin{equation*}
t^{\prime }(\widetilde{q})-c^{\prime }(\widetilde{q})=0.
\end{equation*}%
And, following Lemma \ref{lemm3}, if $t(\widetilde{q})=\mathrm{cav}[t](%
\widetilde{q})$, we then have that: 
\begin{equation*}
\widetilde{q}\in \underset{q\in \lbrack 0,1]}{\arg \max }\ t(q)-c(q);
\end{equation*}%
in other words, in this case, we thus have that $\widetilde{q}=q(c,t).$

We now compute the efficient quality and the respective social surplus when
the cost is \eqref{eq:cost}. The first-order condition is given by: 
\begin{equation*}
u^{\prime }(\overline{q}(c))-c^{\prime }(\overline{q}(c))=0.
\end{equation*}%
We thus have that: 
\begin{equation*}
\overline{q}(c)=\widetilde{\gamma }^{\frac{1}{\sigma -1}}.
\end{equation*}%
Hence, the efficient social surplus is given by: 
\begin{equation}
W(c)=u(\overline{q}(c))-c(\overline{q}(c))=\frac{1-\sigma }{\sigma }\left( 
\frac{1}{\widetilde{\gamma }}\right) ^{\frac{\sigma }{1-\sigma }}+\widetilde{%
\gamma }\widetilde{q}.  \label{cc}
\end{equation}%
%
%
%
%
%
%
%
%
%
%
%
%
%
%
%
%
%
%
%
%
%
%
%
%
%
%
%
%
%
%
%
%
%
%
%
%
%
%
%
%
%
%
%
%
%
%
%
%

We analyze three cases. We first analyze the case in which $Z(q)$ is
strictly decreasing in some part of the domain. We then analyze the case in
which $Z(q)$ is weakly increasing where we distinguish the cases in which: 
\begin{equation*}
\lim_{q\rightarrow \infty }Z(q)\neq z^{\ast }\text{ or }\lim_{q\rightarrow
0}Z(q)\neq z^{\ast }.
\end{equation*}%
We show that in each of the cases the bound is satisfied.

\textbf{(Case 1)} Suppose first that there exists $\widetilde{q}$ such that $%
Z^{\prime }(\widetilde{q})<0$. Following Lemma \ref{lemm3}.3, without loss
of generality we can take $\widetilde{q}$ such that $t(\widetilde{q})=%
\mathrm{cav}[t](\widetilde{q})$, and so (also following Lemma \ref{lemm3})
we have that $\widetilde{q}=q(c,t)$. Using that $Z^{\prime }(\widetilde{q}%
)=z^{\prime }(\widetilde{q})<0,$ \eqref{ccsa} implies that: 
\begin{equation}
\widetilde{\gamma }<\widetilde{z}(\widetilde{q})^{\sigma -1}\leq (\widetilde{%
q})^{\sigma -1}.  \label{vcd}
\end{equation}%
We note that: 
\begin{equation*}
\frac{1-\sigma }{\sigma }\left( \frac{1}{\gamma }\right) ^{\frac{\sigma }{%
1-\sigma }}+\gamma \widetilde{q}
\end{equation*}%
is strictly quasi-convex in $\gamma $ (fixing $\widetilde{q}$) with a unique
minimum at $\gamma =\widetilde{q}^{\sigma -1}$. We thus have that: 
\begin{equation}
W(c)=\frac{1-\sigma }{\sigma }\left( \frac{1}{\widetilde{\gamma }}\right) ^{%
\frac{\sigma }{1-\sigma }}+\widetilde{q}\widetilde{\gamma }>(1-\sigma )\frac{%
\widetilde{q}^{\sigma }}{\sigma }\widetilde{z}^{\frac{\sigma }{\sigma -1}}+%
\widetilde{z}\sigma \frac{\widetilde{q}^{\sigma }}{\sigma }=u(q(c,t))\left( 
\widetilde{z}\sigma +(1-\sigma )\widetilde{z}^{\frac{\sigma }{\sigma -1}%
}\right) .  \label{effw}
\end{equation}%
The inequality comes from the quasi-concavity of \eqref{cc} and \eqref{vcd};
the second equality comes from the fact that $\widetilde{q}=q(c,t)$. We thus
have that 
\begin{equation}
\frac{U(c,t)}{W(c)}=\frac{u(q(c,t))-t(q(c,t))}{u(\overline{q}(c))-c(%
\overline{q}(c))}<\frac{(1-\widetilde{z})}{(1-\sigma )\widetilde{z}^{\frac{%
\sigma }{\sigma -1}}+\widetilde{z}\sigma }\leq \frac{1-z^{\ast }}{(1-\sigma
)(z^{\ast })^{\frac{\sigma }{\sigma -1}}+\sigma z^{\ast }}.  \label{zzxx11}
\end{equation}%
where the first inequality follows from \eqref{effw} while the second
inequality follows from the definition of $z^{\ast }.$

\textbf{(Case 2)} We now consider the case that $t(q)\neq z^{\ast }u(q)$, $%
Z^{\prime }(q)\geq 0$ for all $q$ and 
\begin{equation*}
\lim_{q\rightarrow \infty }Z(q)\neq z^{\ast }.
\end{equation*}%
Since $Z(q)\in \lbrack 0,1]$ we must have that $Z^{\prime }(q)$ must
converge to 0 fast enough as $q\rightarrow \infty $, in fact, 
\begin{equation*}
\underset{q\rightarrow \infty }{\lim \inf }\ Z^{\prime }(q)q=0.
\end{equation*}%
Following Lemma \ref{lemm3}.4, we also have that: 
\begin{equation*}
\underset{q\rightarrow \infty }{\lim \inf }\ z^{\prime }(q)q=0.
\end{equation*}%
We can thus find some $\widetilde{q}$ large enough such that: 
\begin{equation*}
\frac{(1-\widetilde{z})}{(1-\sigma )\left( {\ \tilde{z}+z^{\prime }(%
\widetilde{q})\frac{\widetilde{q}}{\sigma }}\right) ^{\frac{\sigma }{\sigma
-1}}+\left( \tilde{z}+z^{\prime }(\widetilde{q})\frac{\widetilde{q}}{\sigma }%
\right) \sigma }<\frac{1-z^{\ast }}{(1-\sigma )(z^{\ast })^{\frac{\sigma }{%
\sigma -1}}+\sigma z^{\ast }}.
\end{equation*}%
We now fix this $\widetilde{q}$ and prove the result.

Replacing \eqref{ccsa} into \eqref{cc}, we get: 
\begin{align}
u(\overline{q}(c))-c(\overline{q}(c))=& \frac{1-\sigma }{\sigma }\left( 
\frac{1}{\tilde{z}u^{\prime }(\widetilde{q})+z^{\prime }(\widetilde{q})u(%
\widetilde{q})}\right) ^{\frac{\sigma }{1-\sigma }}+q(c,t)\left( \tilde{z}%
u^{\prime }(\widetilde{q})+z^{\prime }(\tilde{q})u(\widetilde{q})\right) 
\notag  \\
=& \frac{1-\sigma }{\sigma }\widetilde{q}^{\sigma }\left( \frac{1}{\tilde{z}%
+z^{\prime }(\widetilde{q})\frac{\widetilde{q}}{\sigma }}\right) ^{\frac{%
\sigma }{1-\sigma }}+\widetilde{q}^{\sigma }\left( \tilde{z}+z^{\prime }(%
\widetilde{q})\frac{\widetilde{q}}{\sigma }\right) .  \label{dcs2}
\end{align}%
We thus have that: 
\begin{equation}
\frac{u(q(c,t))-t(q(c,t))}{u(\overline{q}(c))-c(\overline{q}(c))}=\frac{(1-%
\widetilde{z})}{(1-\sigma )\left( {\ \tilde{z}+z^{\prime }(\widetilde{q})%
\frac{\widetilde{q}}{\sigma }}\right) ^{\frac{\sigma }{\sigma -1}}+\left( 
\tilde{z}+z^{\prime }(\widetilde{q})\frac{\widetilde{q}}{\sigma }\right)
\sigma }<\frac{1-z^{\ast }}{(1-\sigma )(z^{\ast })^{\frac{\sigma }{\sigma -1}%
}+\sigma z^{\ast }}.  \label{cdczz1}
\end{equation}%
We thus prove the result.

\textbf{(Case 3)} We now consider the case that $t(q)\neq z^{\ast }u(q)$, $%
Z^{\prime }(q)\geq 0$ for all $q$ and 
\begin{equation*}
\lim_{q\rightarrow 0}Z(q)\neq z^{\ast }.
\end{equation*}%
We note that we must have that: 
\begin{equation*}
\underset{q\rightarrow 0}{\lim \inf }\ Z^{\prime }(q)q=0.
\end{equation*}%
Following Lemma \ref{lemm3}.4, we also have that: 
\begin{equation*}
\underset{q\rightarrow 0}{\lim \inf }\ z^{\prime }(q)q=0.
\end{equation*}%
We thus find some $\widetilde{q}$ small enough such that: 
\begin{equation*}
\frac{(1-\widetilde{z})}{(1-\sigma )\left( {\ \tilde{z}+z^{\prime }(%
\widetilde{q})\frac{\widetilde{q}}{\sigma }}\right) ^{\frac{\sigma }{\sigma
-1}}+\left( \tilde{z}+z^{\prime }(\widetilde{q})\frac{\widetilde{q}}{\sigma }%
\right) \sigma }<\frac{1-z^{\ast }}{(1-\sigma )(z^{\ast })^{\frac{\sigma }{%
\sigma -1}}+\sigma z^{\ast }}.
\end{equation*}%
We can reach \eqref{dcs2} in the same way as before, which in turn allows us
to also reach \eqref{cdczz1} the same as before. \bigskip

\textbf{(Part II: Attainment of Upper Bound through }$t(q)=z^{\ast }u(q)$%
\textbf{)} We now prove that, if $t(q)=z^{\ast }u(q)$ then 
\begin{equation*}
\frac{U(c,t)}{W(c)}\geq B(\sigma ),
\end{equation*}%
for all cost functions, with equality if and only if the marginal cost
satisfies: 
\begin{equation*}
c^{\prime }(q)\qquad 
\begin{cases}
=0, & q\in \lbrack 0,\widetilde{q}]; \\ 
=z^{\ast }u^{\prime }(\widetilde{q}), & q\in \lbrack \widetilde{q},(z^{\ast
})^{\frac{1}{\sigma -1}}\widetilde{q}]; \\ 
\geq z^{\ast }u^{\prime }(\widetilde{q}), & q\in \lbrack (z^{\ast })^{\frac{1%
}{\sigma -1}}\widetilde{q},\infty ).%
\end{cases}%
\end{equation*}%
for some $\widetilde{q}.$ For any mechanism $t(q)=zu(q)$ for some generic $%
z\in \left[ 0,1\right] $ we get the following. The seller will choose
quantity $q(c,t)$ satisfying: 
\begin{equation*}
z(q(c,t))^{\sigma -1}-c^{\prime }(q(c,t))=0.
\end{equation*}%
The buyer surplus is is a fraction $(1-z^{\ast })$ of the buyer's gross
utility, that is, 
\begin{equation}
U(c,t)=(1-z)\frac{(q(c,t))^{\sigma }}{\sigma }.  \label{eq:cons}
\end{equation}%
The seller surplus is at most: 
\begin{equation}
\Pi (c,t)\leq t(q(c,t))=z\frac{(q(c,t))^{\sigma }}{\sigma }.  \label{pro44}
\end{equation}%
The inequality is satisfied with equality if and only if $c(q)=0$ for all $%
q\in \lbrack 0,q(c,t)]$. The deadweight loss due to inefficiencies can also
be bounded: 
\begin{equation}
DWL=\int_{q(c,t)}^{\overline{q}(c)}\left( u^{\prime }(q)-c^{\prime }\left(
q\right) \right) dq\leq u(\overline{q}(c))-u(q(c,t))-c^{\prime }(q(c,t))(%
\overline{q}(c)-q(c,t)),  \label{eq:dwl}
\end{equation}%
where we recall that $\overline{q}(c)$ is the efficient quantity and $c$ is
convex so $c^{\prime }(q)\geq c^{\prime }(q(c,t))$ for all $q\in \lbrack
q(c,t),\overline{q}(c)]$. The efficient quantity $\overline{q}(c)$\ can be
bounded from above as follows: 
\begin{equation*}
\overline{q}(c)=\left( c^{\prime }(\overline{q}(c))\right) ^{\frac{1}{\sigma
-1}}\leq \left( c^{\prime }(q(c,t))\right) ^{\frac{1}{\sigma -1}}=z^{\frac{1%
}{\sigma -1}}q(c,t).
\end{equation*}%
We thus have that: 
\begin{equation*}
DWL\leq u(z^{\frac{1}{\sigma -1}}q(c,t))-u(q(c,t))-c^{\prime }(q(c,t))(%
\overline{q}(c)-q(c,t))\leq u(q(c,t))\left( (1-\sigma )z^{\frac{\sigma }{%
\sigma -1}}+\sigma z-1\right) .
\end{equation*}%
Furthermore, this inequality is satisfied with equality if and only if $%
c^{\prime }(q)$ is constant for all $q\in \lbrack q(c,t),\overline{q}(c)].$
The ratio (\ref{ratio}):%
\begin{equation*}
\frac{U}{W}=\frac{U}{U+\Pi +DWL},
\end{equation*}%
after inserting the above terms (\ref{eq:cons})-(\ref{eq:dwl}) satisfies 
\begin{equation}
\frac{U}{U+\Pi +DWL}\geq \frac{1-z}{(1-\sigma )z^{\frac{\sigma }{\sigma -1}%
}+\sigma z}.  \label{dcc2}
\end{equation}%
Thus, for any $t(q)=zu(q)$ and any convex cost function, we obtain the
inequality (\ref{dcc2}). Now, replacing $z$ with $z^{\ast }$ we get a lower
bound of which we know by the first part of the proof that it is also an
upper bound. Finally, we note that \eqref{dcc2} is satisfied with equality
if and only if $c(q)=0$ for all $q\in \lbrack 0,q(c,t)]$ and $c^{\prime }(q)$
is constant for all $q\in \lbrack q(c,t),\overline{q}(c)].$ This concludes
the proof.
\end{proof}

\begin{proof}[Proof of Lemma \protect\ref{lem:ctc} and Theorem \protect\ref%
{thm:rev}]
The theorem is proved by performing a change of variables to recover our
original power utility model. For the given utility function $u\left(
q\right) $, we define: 
\begin{equation*}
x=\left( {\widehat{\sigma }u(q)}\right) ^{\frac{1}{\widehat{\sigma }}},
\end{equation*}%
with $\widehat{\sigma }$ defined as in \eqref{eq:sig} (the relevant case is
when $\underline{\delta }>0$). To make the notation more compact, throughout
the proof we write $\sigma $ instead of $\widehat{\sigma }.$ We then have
that the utility of the buyer and the cost of the seller in terms of the
variable $x$ is: 
\begin{equation*}
\widehat{u}(x)=\frac{x^{\sigma }}{\sigma }\text{ and }\widehat{c}(x)=c\left(
u^{-1}\left( \frac{x^{\sigma }}{\sigma }\right) \right) .
\end{equation*}%
We then prove Theorem \ref{thm:rev} by applying Theorem \ref{thm:osg}. The
proof reduces to showing that $\widehat{c}$ is increasing and convex.

We first verify that $\widehat{c}$ is always increasing: 
\begin{equation*}
\widehat{c}^{\prime }(x)=c^{\prime }\left( u^{-1}\left( \frac{x^{\sigma }}{%
\sigma }\right) \right) u^{-1^{\prime }}\left( \frac{x^{\sigma }}{\sigma }%
\right) x^{\sigma -1}\geq 0.
\end{equation*}%
For the second derivative, we get: 
\begin{eqnarray*}
\widehat{c}^{\prime \prime }(x) &=&\frac{\sigma }{x^{2-\sigma }}c^{\prime
}\left( u^{-1}\left( \frac{x^{\sigma }}{\sigma }\right) \right) {%
u^{-1^{\prime }}\left( \frac{x^{\sigma }}{\sigma }\right) }\times \\
&&\left( \frac{c^{\prime \prime }\left( u^{-1}\left( \frac{x^{\sigma }}{%
\sigma }\right) \right) u^{-1}\left( \frac{x^{\sigma }}{\sigma }\right) }{%
c^{\prime }\left( u^{-1}\left( \frac{x^{\sigma }}{\sigma }\right) \right) }%
\frac{u^{-1^{\prime }}\left( \frac{x^{\sigma }}{\sigma }\right) \frac{%
x^{\sigma }}{\sigma }}{u^{-1}\left( \frac{x^{\sigma }}{\sigma }\right) }+%
\frac{x^{\sigma }}{\sigma }\frac{u^{-1^{\prime \prime }}\left( \frac{%
x^{\sigma }}{\sigma }\right) }{u^{-1^{\prime }}\left( \frac{x^{\sigma }}{%
\sigma }\right) }-\frac{(1-\sigma )}{\sigma }\right) .
\end{eqnarray*}%
We now note that:%
\begin{equation*}
\frac{c^{\prime \prime }\left( u^{-1}\left( \frac{x^{\sigma }}{\sigma }%
\right) \right) u^{-1}\left( \frac{x^{\sigma }}{\sigma }\right) }{c^{\prime
}\left( u^{-1}\left( \frac{x^{\sigma }}{\sigma }\right) \right) }\frac{%
u^{-1^{\prime }}\left( \frac{x^{\sigma }}{\sigma }\right) \frac{x^{\sigma }}{%
\sigma }}{u^{-1}\left( \frac{x^{\sigma }}{\sigma }\right) }=\frac{qc^{\prime
\prime }\left( q\right) }{c^{\prime }\left( q\right) }\frac{u\left( q\right) 
}{qu\left( q\right) }
\end{equation*}%
and 
\begin{equation*}
\frac{u^{-1^{\prime }}\left( \frac{x^{\sigma }}{\sigma }\right) }{%
u^{-1^{\prime \prime }}\left( \frac{x^{\sigma }}{\sigma }\right) }\frac{%
\sigma }{x^{\sigma }}=\frac{u^{-1^{\prime }}\left( u(q)\right) }{%
u^{-1^{\prime \prime }}\left( u(q)\right) }\frac{1}{u(q)}=-\frac{u^{\prime
}(q)}{qu^{\prime \prime }(q)}\cdot \frac{qu^{\prime }(q)}{u(q)}.
\end{equation*}%
We can thus write the second derivative as follows: 
\begin{equation*}
\widehat{c}^{\prime \prime }(x)=\frac{\sigma }{x^{2-\sigma }}c^{\prime
}\left( u^{-1}\left( \frac{x^{\sigma }}{\sigma }\right) \right) {%
u^{-1^{\prime }}\left( \frac{x^{\sigma }}{\sigma }\right) }{}\left( \delta
(q,c)-\frac{(1-\sigma )}{\sigma }\right) ,
\end{equation*}%
where we use $\delta \left( c,q\right) $ as defined in (\ref{ccd00}). Since $%
(1-\sigma )/\sigma =\underline{\delta }$, we have that $\widehat{c}^{\prime
\prime }(x)\geq 0$. Applying Theorem \ref{thm:osg}, we get the result.
\end{proof}

\medskip

\begin{proof}[Proof of Theorem \protect\ref{thm:wel}]
The proof of Theorem \ref{thm:osg} goes through without any significant
difference. We now have that: 
\begin{equation*}
u(q(c,t))-t(q(c,t))+\alpha \left( t(q(c,t))-c(q(c,t))\right) =u(\widetilde{q}%
)(1-\widetilde z+\alpha \widetilde z),
\end{equation*}%
where the cost continues to be as in \eqref{eq:cost}. We now explain how
Parts I and II of the proof need to be modified to account for the weight on
profits.

For Part I of the proof, we just have to add $\alpha z$ to the numerator in %
\eqref{zzxx11}, and analogously to Cases 2 and 3 analyzed therein. Note that
here we provide a cost and show that the competitive ratio always satisfies
the corresponding bound when we add the profits (weighted by $\alpha $) to
the numerator.

For Part II of the proof, we first note that \eqref{ratio} is at least equal
to $\alpha $. To prove this, we note that when $t(q)=u(q)$, the seller
chooses the efficient quality and extracts the efficient social surplus, so
we get that \eqref{ratio} is given by: 
\begin{equation*}
\frac{u(q(c,t))-t(q(c,t))+\alpha (t(q(c,t))-c(q(c,t)))}{u(\overline{q}(c))-c(%
\overline{q}(c))}=\alpha .
\end{equation*}%
We also note that, for any given mechanism $t$, the efficient surplus is the
sum of buyer surplus, profits and deadweight loss: 
\begin{equation*}
u(\overline{q}(c))-c(\overline{q}(c))=U+\Pi +DWL.
\end{equation*}%
Since, the ratio is greater than $\alpha $, we have that: 
\begin{equation*}
\frac{U+\alpha \Pi }{U+\Pi +DWL}
\end{equation*}%
is decreasing in $\Pi $. We thus have that \eqref{pro44} continues to be the
relevant bound. The rest of the proof remains the same (by adding $\alpha z$
to the numerator).
\end{proof}

\medskip

\begin{proof}[Proof of Proposition \protect\ref{prop:seller}]
The proof follows the same way as the analysis in Section \ref{sec:lc} with
the change of variables $c=1/v.$ Since in Section \ref{sec:lc} the constant
share mechanism is found to be the Bayesian-optimal mechanism when the
distribution of cost is a power function, for the optimal selling mechanism,
the constant markup mechanism will be the Bayesian-optimal mechanism when
the distribution of values $v=1/c$ follows a Pareto distribution (the
reciprocal of a random variable that follows a power distribution follows a
Pareto distribution). The rest of the analysis follows the same way.
\end{proof}

\medskip \newpage

\section{Supplemental Appendix: Additional Material}

\label{appB}

Our main result in the nonlinear cost environment, Theorem \ref{thm:osg},
provides the competitive ratio for deterministic mechanisms, which is
attained by a constant share mechanism. We now relax the restriction to
deterministic policies and, in consequence, obtain an improved competitive
ratio. Yet, every policy in the support of the stochastic mechanism remains
a constant share rule. And surprisingly, the relaxation to stochastic
policies does not yield a noticeable improvement in the competitive ratio.

In a second relaxation, we weaken the benchmark--the denominator--in the
competitive ratio. So far, the benchmark in the competitive ratio was the
efficient social surplus. We weaken the benchmark and compare the
performance of the optimal constant share mechanism with the revenue of the
Bayes optimal (unconstrained) mechanism. We find that the constant surplus
sharing in fact attains the same bound as against the social surplus. This
section thus provides additional support to the notion that simple rules
perform well in the procurement environment.

\paragraph{Stochastic Policies}

We now allow for randomizations over transfer policies, thus allowing the
buyer to choose a distribution $F$ over transfer policies, thus: $F\in
\Delta \{t\left\vert t:\mathbb{R}\rightarrow \mathbb{R}\right. \}$. We thus
seek to identify a competitive ratio: 
\begin{equation*}
\underline{B}(\sigma )\triangleq \max_{F\in \Delta \{t:t:\mathbb{R}%
\rightarrow \mathbb{R}\}}\min_{c\in C_{cx}}\frac{\int U(c,t)dF(t)}{W(c)}.
\end{equation*}%
As we allow for stochastic policies, we can also ask whether the solution
from the related $\min \max $ problem: 
\begin{equation*}
\overline{B}(\sigma )\triangleq \min_{G\in \Delta C_{cx}}\max_{\{t:\mathbb{R}%
\rightarrow \mathbb{R}\}}\frac{\int U(c,t)dG(c)}{\int W(c)dG(c)}
\end{equation*}%
coincides with the $\max \min $ value, that is whether $\underline{B}(\sigma
)=\overline{B}(\sigma )$. A \emph{saddle point} is a pair of distributions $%
(F^{\ast },G^{\ast })$ such that: 
\begin{equation}
\min_{G\in \Delta C_{cx}}\frac{\int \int U(c,t)dG(c)dF^{\ast }(t)}{\int
W(c)dG(c)}=\max_{F\in \Delta \{t:t:\mathbb{R}\rightarrow \mathbb{R}\}}\frac{%
\int \int U(c,t)dG^{\ast }(c)dF(t)}{\int W(c)dG^{\ast }(c)}.  \label{eq:sa}
\end{equation}%
That is, these distributions attain the values of the saddle point. We
answer both of these questions in the affirmative in the next result. To
make the notation more compact, we define: 
\begin{equation*}
\widehat{B}(\sigma )\triangleq \frac{1}{\sigma ^{\frac{-\sigma }{1-\sigma }%
}-\sigma -\log (1-\sigma )},
\end{equation*}%
which will be the value of the saddle point.

\begin{proposition}[Saddle Point\ Property with Stochastic Policies]
\label{prop:saddle}$\quad $\newline
A saddle point (\ref{eq:sa}) exists and the competitive ratio is: 
\begin{equation*}
\underline{B}(\sigma )=\overline{B}(\sigma )=\widehat{B}\left( \sigma
\right) .
\end{equation*}%
The saddle point is attained by distributions $(F^{\ast },G^{\ast })$ that
only place positive weights on constant share mechanisms.
\end{proposition}

\begin{proof}
To find the saddle point, we consider a parametrized class of cost functions
of the following form: 
\begin{equation*}
c_{\widetilde{\gamma }}(q)=%
\begin{cases}
0, & \text{if }q\in \lbrack 0,q_{0}]; \\ 
u(q)-u(q_{0})-\kappa \log (\frac{q}{q_{0}}), & \text{if }q\in \lbrack
q_{0},q_{1}]; \\ 
u(q_{1})-u(q_{0})-\kappa \log (\frac{q_{1}}{q_{0}})+\widetilde{\gamma }%
(q-q_{1}), & \text{if }q\geq q_{1},%
\end{cases}%
\end{equation*}%
The cost functions are parametrized by $\widetilde{\gamma }$ and the rest of
the parameters are determined as follows: 
\begin{equation*}
{q_{0}}\triangleq (1-\sigma )^{\frac{1}{\sigma }}\left( \frac{\sigma }{%
\widetilde{\gamma }}\right) ^{\frac{1}{1-\sigma }}\text{ ; \ \ \ }{q_{1}}%
\triangleq \left( \frac{\sigma }{\widetilde{\gamma }}\right) ^{\frac{1}{%
1-\sigma }}\text{ ; \ \ }\kappa \triangleq (1-\sigma )\left( \frac{\sigma }{%
\widetilde{\gamma }}\right) ^{\frac{\sigma }{1-\sigma }}.
\end{equation*}%
We illustrate this class of cost functions in Figure \ref{costt22} by
showing the marginal cost $c_{\widetilde{\gamma }}^{\prime }(q)$. We can see
that the marginal cost is 0 for all $q\in \lbrack 0,q_{0}]$, then it is
increasing for all $q\in \lbrack q_{0},q_{1}]$ (hence $c_{\widetilde{\gamma }%
}$ is strictly convex in this range), and finally the marginal cost is equal
to $\widetilde{\gamma }$ for all $q\in \lbrack q_{1},\infty ).$ We now show
that this class of cost functions and the constant share mechanisms form a
saddle point (with the appropriate mixing).
\begin{figure}[t]
	\centering
	\includegraphics[width=4.2457in,height=2.6442in,keepaspectratio]{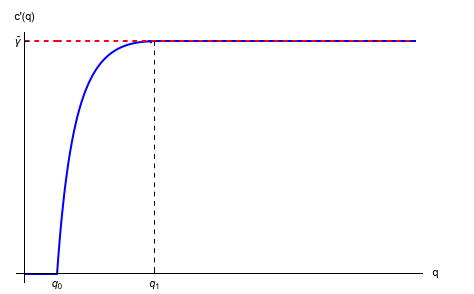}
	\caption{Illustration of Marginal Cost Function $c_{\gamma}'(q)$}
	\label{costt22}
\end{figure}
We prove the result in two steps,
first prove the maximin problem attains $\widehat{B}$ and then prove that
the minmax problem also attains $\widehat{B}$.

\textbf{(Step 1: $\widehat{B}(\sigma )=\underline{B}(\sigma )$)} Suppose the
buyer uses a payment rule: 
\begin{equation*}
t(q)=zu(q),
\end{equation*}%
but randomizes $z$ according to distribution 
\begin{equation}
F(z)=%
\begin{cases}
-\widehat{B}(\sigma )\log (1-z), & \text{if }z<\sigma ; \\ 
1, & \text{if }z\geq \sigma .%
\end{cases}
\label{ccaz}
\end{equation}%
Hence, $F$ is absolutely continuous in $[0,\sigma )$ and has an atom at $%
z=\sigma .$

We first note that, for any linear cost $c(q)=\gamma q$, the efficient
social surplus and buyer surplus are: 
\begin{equation*}
W(c)=\frac{1-\sigma }{\sigma }\left( \frac{1}{\gamma }\right) ^{\frac{\sigma 
}{1-\sigma }}\text{ and \ \ }U(q,z)=\frac{(1-z)}{\sigma }\left( \frac{z}{%
\gamma }\right) ^{\frac{\sigma }{1-\sigma }}.
\end{equation*}%
We thus have that for any linear cost function $c(q)=\gamma q$ the expected
competitive ratio is: 
\begin{equation*}
\frac{\int_{0}^{\sigma }U(c,z)dF(z)}{W(c)}=\widehat{B}(\sigma ).
\end{equation*}%
We remark that to compute the integral, we need to account for the fact that 
$F$ has an atom at $z=\sigma .$ We now show that the linear cost function
minimizes the competitive ratio: 
\begin{equation*}
\widehat{B}(\sigma )=\min_{c}\frac{\int_{0}^{\sigma }U(c,z)dF(z)}{W(c)},
\end{equation*}%
hence the proposed randomization over constant share mechanisms indeed
maximizes the competitive ratio.

We denote by $\mathcal{C}_{b}$ the set of convex cost functions that have
constant marginal cost for high enough $q$: 
\begin{equation*}
\mathcal{C}_{b}\triangleq \{c:c\text{ is increasing, convex, and }c^{\prime
\prime }(q)=0\text{ for all }q\text{ such that }c^{\prime }(q)\geq \sigma
u^{\prime }(q)\}.
\end{equation*}%
In other words, $\mathcal{C}_{b}$ consists of increasing and convex
functions with constant marginal cost whenever the marginal cost exceeds $%
\sigma u^{\prime }(q)$. Given that the distribution over $z$ is given by %
\eqref{ccaz} it is without loss of generality to consider cost functions in $%
\mathcal{C}_{b}$. More precisely, for all $c\not\in C_{b}$, there exists $%
\widehat{c}\in C_{b}$ such that: 
\begin{equation*}
U(c,z)=U(\widehat{c},z)\text{ and }W(c,z)\leq W(\widehat{c},z),
\end{equation*}%
for all $z\in \lbrack 0,\sigma ].$ The inequality comes from the fact that
lowering the marginal cost of units that are never sold to the buyer always
increases the efficient social welfare without changing the buyer surplus.

We now fix $c\in \mathcal{C}_{b}$ and prove that: 
\begin{equation*}
\frac{\int_{0}^{\sigma }U(c,z)dF(z)}{W(c)}=\widehat{B}(\sigma ).
\end{equation*}%
For this, we consider a class of cost functions $c_{x}(q)$ parametrized by ${%
x}\in \lbrack 0,1]$ with the following properties: $(i)$ $c_{0}^{\prime
}(q)=c^{\prime }(\overline{q}(c))$ for all $q\in \mathbb{R}$, (ii) $%
c_{1}(q)=c(q)$ for all $q\in \mathbb{R}$, and (iii) $c_{x}^{\prime }(q)$
strictly decreasing in $q$ and ${x}$, and (iv) continuous in ${x}\in \lbrack
0,1]$. The first property states that $c_{0}(q)$ is a linear cost function
with marginal cost equal to the marginal cost of $c$ at the efficient $q$.
The second property states that at ${x}=1$ the parametrized function
coincides with $c$. The third property states how $c_{x}$ changes with ${x}.$
Since $c\in \mathcal{C}_{b}$, we have that for all $q\geq q(c,\sigma )$, $%
c_{x}^{\prime }(q)$ is constant in ${x}$. To make the notation more compact
we define: 
\begin{equation*}
r({x})\triangleq \frac{\int U(c_{x},z)dF(z)}{W(c_{x})}.
\end{equation*}%
We note that: 
\begin{equation*}
\frac{\partial r({x})}{\partial {x}}=\frac{1}{W(c_{x})}(\frac{\partial \int
U(c_{x},z)dF(z)}{\partial {x}}-r({x})\frac{\partial W(c_{x})}{\partial {x}}).
\end{equation*}%
We now prove this derivative is equal to 0. Finally, to make the notation
more compact, we denote $q({x},z)\triangleq q(c_{x},z)$.

The total surplus is given by: 
\begin{equation*}
W(c_{x})=\int_{0}^{\infty }\max \{u^{\prime }(q)-c_{x}^{\prime }(q),0\}dq.
\end{equation*}%
We compute the derivative of total surplus with respect to ${x}$: 
\begin{equation*}
\frac{\partial W(c_{x})}{\partial {x}}=-\int_{q({x},0)}^{q({x},\sigma )}%
\frac{\partial c_{x}^{\prime }(q)}{\partial {x}}dq.
\end{equation*}%
We now note that: 
\begin{equation*}
\frac{\partial q({x},z)}{\partial z}=-u^{\prime }(q({x},z))\frac{1}{%
zu^{\prime \prime }(q({x},z))-c_{x}^{\prime \prime }(q({x},z))}.
\end{equation*}%
Thus, 
\begin{equation*}
\frac{\partial W(c_{x})}{\partial {x}}=\int_{0}^{\sigma }\frac{\partial
c_{x}^{\prime }(q)}{\partial {x}}\frac{u^{\prime }(q({x},z))}{zu^{\prime
\prime }(q({x},z))-c_{x}^{\prime \prime }(q({x},z))}dz.
\end{equation*}%
We now compute the derivative of the buyer surplus with respect to ${x}.$

The expected buyer surplus: 
\begin{equation*}
\int_{0}^{1}U({x},z)F(z)dz=\widehat{B}(\sigma )\int_{0}^{\sigma }u(q({x}%
,z))dz+(1+\log (1-\sigma )\widehat{B}(\sigma ))(1-\sigma )u(q({x},\sigma )).
\end{equation*}%
We compute the derivative with respect to ${x}$: 
\begin{equation*}
\frac{\partial }{\partial {x}}\int U({x},z)f(z)dz=\widehat{B}(\sigma
)\int_{0}^{\sigma }u^{\prime }(q({x},z))\frac{\partial q({x},z)}{\partial {x}%
}dz.
\end{equation*}%
Note that $q({x},\sigma )$ does not change with ${x}$ so the atom $z=\sigma $
does not appear when computing the derivative (recall that for all $q\geq
q(c,\sigma )$, $c_{x}^{\prime }(q)$ is constant in ${x}$). We now note that: 
\begin{equation*}
\frac{\partial q({x},z)}{\partial {x}}=\frac{\partial c_{x}^{\prime }(q({x}%
,z))}{\partial {x}}\frac{1}{zu^{\prime \prime }(q({x},z))-c_{x}^{\prime
\prime }(q({x},z))}.
\end{equation*}%
We thus get that: 
\begin{equation*}
\frac{\partial }{\partial {x}}\int U({x},z)f(z)dz=\widehat{B}(\sigma
)\int_{0}^{\sigma }\frac{\partial c_{x}^{\prime }(q({x},z))}{\partial {x}}%
\frac{u^{\prime }(q({x},z))}{zu^{\prime \prime }(q({x},z))-c_{x}^{\prime
\prime }(q({x},z))}dz.
\end{equation*}%
Hence, we have that: 
\begin{equation*}
\frac{\partial }{\partial {x}}\int U({x},z)f(z)dz=\widehat{B}(\sigma )\frac{%
\partial W(c_{x})}{\partial {x}},
\end{equation*}%
which we can use to compute the derivative of the competitive ratio.

We have that: 
\begin{equation*}
\frac{\partial r({x})}{\partial {x}}=\frac{\partial W(c_{x})}{\partial {x}}%
\frac{1}{W(c)}(\widehat{B}(\sigma )-r({x})).
\end{equation*}%
Hence $r({x})$ is determined by a linear differential equation. Since $r({x}%
)=\widehat{B}(\sigma )$ is a solution to this differential equation; this is
indeed the unique solution. Hence, we obtain the result.

\textbf{(Step 2: $\widehat{B}(\sigma )=\overline{B}(\sigma )$)} We consider
a distribution over cost functions $c_{\widetilde{\gamma }}$ according to a
cumulative distribution 
\begin{equation*}
F(\widetilde{\gamma })=\widetilde{\gamma }^{\theta },
\end{equation*}%
for some $\theta >\sigma /(1-\sigma ).$

Following standard techniques the buyer surplus is given by: 
\begin{equation*}
U(c,t)=\int_{0}^{1}\max_{q}\left\{ \left( u(q)-c_{\widetilde{\gamma }}(q)-%
\frac{F(\widetilde{\gamma })}{f(\widetilde{\gamma })}\frac{\partial c_{%
\widetilde{\gamma }}(q)}{\partial \widetilde{\gamma }}\right) \right\} dF(%
\widetilde{\gamma }).
\end{equation*}%
To make the notation more compact, we define: 
\begin{equation*}
V_{\widetilde{\gamma }}(q)\triangleq \left( u(q)-c_{\widetilde{\gamma }}(q)-%
\frac{F(\widetilde{\gamma })}{f(\widetilde{\gamma })}\frac{\partial c_{%
\widetilde{\gamma }}(q)}{\partial \widetilde{\gamma }}\right) ,
\end{equation*}%
which is the virtual surplus. Computing explicitly, we get: 
\begin{equation*}
V_{\widetilde{\gamma }}(q)=%
\begin{cases}
0, & \text{if }q\in \lbrack 0,q_{0}]; \\ 
\frac{\left( \frac{\sigma }{\widetilde{\gamma }}\right) ^{\frac{\sigma }{%
1-\sigma }}\left( \theta (1-\sigma )+{(\theta -\frac{\sigma }{1-\sigma }%
)\left( \sigma (1-\sigma )\log (q)+\sigma \log \left( \widetilde{\gamma }%
\frac{1-\sigma }{\sigma }\right) -\log (1-\sigma )\right) }\right) }{\theta
\sigma }, & \text{if }q\in \lbrack q_{0},q_{1}]; \\ 
{(\sigma \frac{1+\theta }{\theta }-1)\left( \frac{\sigma }{\widetilde{\gamma 
}}\right) ^{\frac{\sigma }{1-\sigma }}}\left( 1+\frac{\log (1-\sigma )}{%
\sigma }\right) -\frac{(\theta +1)}{\theta }\widetilde{\gamma }q+\frac{%
q^{\sigma }}{\sigma }, & \text{if }q\geq q_{1}.%
\end{cases}%
\end{equation*}%
We note that the optimum is at 
\begin{equation*}
\left( \frac{\theta }{\widetilde{\gamma }(1+\theta )}\right) ^{\frac{1}{%
1-\sigma }}=\arg \max_{q\in \mathbb{R}}V_{\widetilde{\gamma }}(q).
\end{equation*}%
We thus have that: 
\begin{equation*}
\max_{q\in \mathbb{R}}V_{\widetilde{\gamma }}(q)={(\sigma \frac{1+\theta }{%
\theta }-1)\left( \frac{\sigma }{\widetilde{\gamma }}\right) ^{\frac{\sigma 
}{1-\sigma }}}\left( 1+\frac{\log (1-\sigma )}{\sigma }\right) +\frac{%
(1-\sigma )}{\sigma }\left( \frac{\theta }{\widetilde{\gamma }(1+\theta )}%
\right) ^{\frac{\sigma }{1-\sigma }},
\end{equation*}%
and 
\begin{equation*}
\max_{t}\ U(c,t)=\left( {(\sigma \frac{1+\theta }{\theta }-1)\left( {\sigma }%
\right) ^{\frac{\sigma }{1-\sigma }}}\left( 1+\frac{\log (1-\sigma )}{\sigma 
}\right) +\frac{(1-\sigma )}{\sigma }\left( \frac{\theta }{(1+\theta )}%
\right) ^{\frac{\sigma }{1-\sigma }}\right) \left( \frac{\theta }{\theta -%
\frac{\sigma }{1-\sigma }}\right) .
\end{equation*}

Computing now the total surplus, we get: 
\begin{equation*}
W(c,t)=\int_{0}^{1}\max_{q}\left\{ u(q)-c_{\widetilde{\gamma }}(q)\right\}
dF(\widetilde{\gamma }).
\end{equation*}%
The efficient social surplus is: 
\begin{equation*}
\overline{q}(\widetilde{\gamma })=\left( \frac{1}{\widetilde{\gamma }}%
\right) ^{\frac{1}{1-\sigma }}.
\end{equation*}%
Thus, the total surplus is given by: 
\begin{equation*}
W=\left( {(\sigma -1)\left( {\sigma }\right) ^{\frac{\sigma }{1-\sigma }}}%
\left( 1+\frac{\log (1-\sigma )}{\sigma }\right) +\frac{(1-\sigma )}{\sigma }%
\right) \left( \frac{\theta }{\theta -\frac{\sigma }{1-\sigma }}\right) .
\end{equation*}%
We thus have that the ratio is given by: 
\begin{equation*}
\max_{t}\frac{U(c,t)}{W(c)}=\frac{\left( {(\sigma \frac{1+\theta }{\theta }%
-1)\left( {\sigma }\right) ^{\frac{\sigma }{1-\sigma }}}\left( 1+\frac{\log
(1-\sigma )}{\sigma }\right) +\frac{(1-\sigma )}{\sigma }\left( \frac{\theta 
}{(1+\theta )}\right) ^{\frac{\sigma }{1-\sigma }}\right) }{\left( {(\sigma
-1)\left( {\sigma }\right) ^{\frac{\sigma }{1-\sigma }}}\left( 1+\frac{\log
(1-\sigma )}{\sigma }\right) +\frac{(1-\sigma )}{\sigma }\right) }.
\end{equation*}%
We now note that: 
\begin{equation*}
\lim_{\theta \downarrow \frac{1-\sigma }{\sigma }}\frac{W}{U}=\frac{1}{%
\left( {-\sigma }-\log (1-\sigma )+\sigma ^{\frac{-\sigma }{1-\sigma }%
}\right) }=\widehat{B}(\sigma ).
\end{equation*}%
We thus prove the result.
\end{proof}

With stochastic policies, the class of critical cost functions becomes
larger. It now consists of cost functions that have three terms, linear
terms, log terms, and utility terms. We illustrate the value attained by the
saddle point in Figure \ref{boundss}. We can see that the value of the
saddle point is quite similar to the bound we found with deterministic
policies $B(\sigma ):$ they are both equal to 1 when $\sigma $ is small and
they both converge to 0 when $\sigma $ converges to 1. The only part where
both values are significantly different is near 1 because in the limit $%
\sigma \rightarrow 1$, we have that $\widehat{B}^{\prime }(\sigma
)\rightarrow -\infty $, while $B^{\prime }(\sigma )$ remains bounded.
\begin{figure}[t]
	\centering
	\includegraphics[width=6.0166in,height=3.1557in,keepaspectratio]{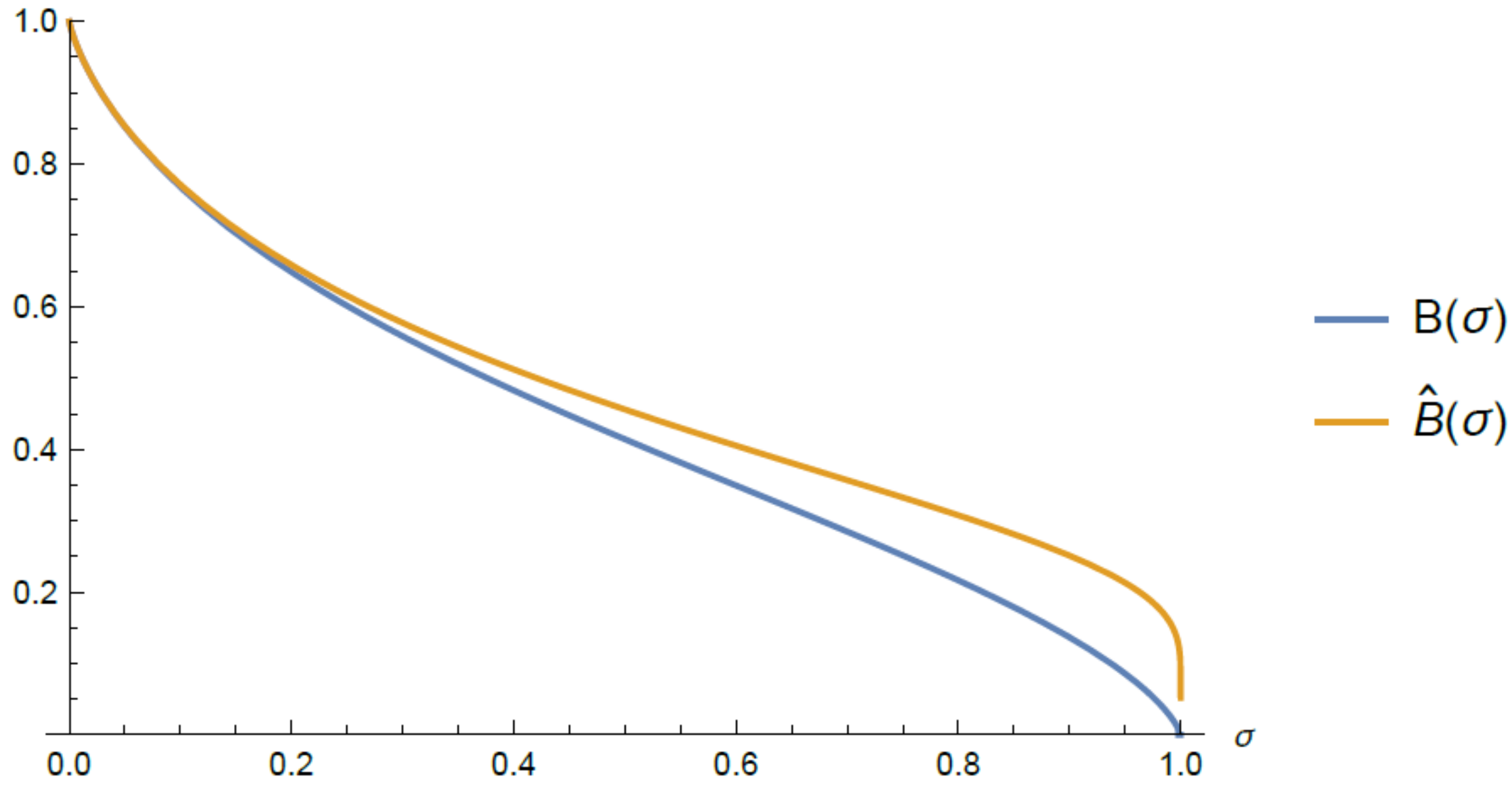}
	\caption{Competitive Ratio Guarantee with Deterministic (blue) and Stochastic (yellow) Transfer Policies}
	\label{boundss}
\end{figure}


\paragraph{Simple vs Bayes Optimal Mechanisms\label{subsec:be}}

Suppose now the buyer had Bayesian uncertainty about the seller's cost, but
the buyer was constrained to using a constant share mechanism. So far, we
have compared the performance of the simple rule relative to the efficient
social welfare. We now compare the performance of the constant share
mechanisms relative to the Bayes-optimal rule. We interpret this constraint
as capturing the fact that computing the optimal transfer when the Bayesian
prior is over all possible convex cost functions may be an intractable
problem.

We recall the basic definitions introduced in Section \ref{sec:bopt}. For
any $G\in \Delta (\mathcal{C}_{cx})$ we define: 
\begin{equation*}
U^{\ast }(G)\triangleq \sup_{t}\int U(c,t)dG(c).
\end{equation*}%
That is, $U^{\ast }(G)$ is the buyer surplus generated by the Bayes-optimal
mechanism when the distribution over cost is $G$. The performance of the
simple rule relative to the Bayes-optimal rule is defined as follows: 
\begin{equation}
\min_{G\in \Delta (\mathcal{C}_{cx})}\max_{z\in \lbrack 0,1]}\frac{\int
U(c,z)dG(c)}{U^{\ast }(G)}.  \label{fff}
\end{equation}%
Thus, we compute the performance of the simple rule relative to the
Bayes-optimal rule. We emphasize that in this case, we allow for the
constant share mechanism to be adapted to the distribution over cost.

\begin{proposition}[Using the Bayesian-Optimal Mechanism as Benchmark]
\label{prop:de}$\quad $\newline
The minmax problem \eqref{fff} attains a value equal to the saddle point: 
\begin{equation*}
\widehat{B}(\sigma )=\min_{G\in \Delta (\mathcal{C}_{cx})}\max_{z\in \lbrack
0,1]}\frac{\int U(c,z)dG(c)}{U^{\ast }(G)}.
\end{equation*}
\end{proposition}

\begin{proof}
To prove the result, we first note that: 
\begin{equation}
\widetilde{B}(\sigma )\geq \min_{G\in \Delta (\mathcal{C}_{cx})}\max_{z\in
\lbrack 0,1]}\frac{\int U(c,z)dG(c)}{\int W(c)dG(c)}=\underline{B}(\sigma ).
\label{eqdc}
\end{equation}%
The inequality follows from the fact that we are replacing the term in the
denominator with a strictly larger term; the equality follows from the fact
that: 
\begin{equation*}
\min_{G\in \Delta (\mathcal{C}_{cx})}\max_{z\in \lbrack 0,1]}\frac{\int
U(c,z)dG(c)}{\int W(c)dG(c)}=\min_{G\in \Delta (\mathcal{C}_{cx})}\max_{F\in
\Delta \lbrack 0,1]}\frac{\int U(c,z)dG(c)dF(z)}{\int W(c)dG(c)}.
\end{equation*}%
That is, by allowing randomizations over constant share mechanisms we do not
change the value of the minmax problem. We now note that: 
\begin{equation*}
\min_{G\in \Delta (\mathcal{C}_{cx})}\max_{F\in \Delta \lbrack 0,1]}\frac{%
\int U(c,z)dG(c)dF(z)}{\int W(c)dG(c)}\leq \min_{G\in \Delta (\mathcal{C}%
_{cx})}\max_{F\in \Delta \{t:t:\mathbb{R}\rightarrow \mathbb{R}\}}\frac{\int
U(c,t)dG(c)dF(t)}{\int W(c)dG(c)}=\underline{B}(\sigma ).
\end{equation*}%
The inequality follows from the fact that by minimizing over all transfers
rather than constant share mechanisms we get a weakly larger minimization;
the equality follows from the fact that this is the saddle point
characterized in Proposition \ref{prop:saddle}. Following Proposition \ref%
{prop:saddle}, by randomizing over constant share mechanisms the buyer can
guarantee himself $\underline{B}(\sigma )$, so 
\begin{equation*}
\min_{G\in \Delta (\mathcal{C}_{cx})}\max_{F\in \Delta \lbrack 0,1]}\frac{%
\int U(c,z)dG(c)dF(z)}{\int W(c)dG(c)}\geq \underline{B}(\sigma ).
\end{equation*}%
We thus get \eqref{eqdc}. We now note that: 
\begin{equation}
\widetilde{B}(\sigma )\leq \min_{c\in \mathcal{C}_{cx}}\max_{z\in \lbrack
0,1]}\frac{U(c,z)}{W(c)}\leq \overline{B}(\sigma ).  \label{eqdc2}
\end{equation}%
The first inequality follows from the fact that, in contrast to \eqref{fff},
we allow the minimization to be only over deterministic cost functions
(rather than distributions over cost). In the denominator we replaced $%
U^{\ast }(F)$ with $W(c)$ but this is without loss of generality because for
any degenerate distribution $F$ we have that 
\begin{equation*}
U^{\ast }(F)=W(c).
\end{equation*}%
That is, if there is no uncertainty about cost, the optimal mechanism
extracts the efficient social surplus. The second inequality in \eqref{eqdc2}
follows from the fact that for any $\widetilde{\gamma }\in \mathbb{R}$ and $%
z\in \lbrack 0,\sigma ]$: 
\begin{equation*}
\frac{U(c_{\widetilde{\gamma }},z)}{W(c_{\widetilde{\gamma }})}=\overline{B}%
(\sigma ).
\end{equation*}%
Hence, by minimizing over all costs we get a weakly smaller number.
Following \eqref{eqdc}-\eqref{eqdc2} and Proposition \ref{prop:saddle}, we
obtain the result.
\end{proof}

We thus obtain the same competitive ratio as in Proposition \ref{prop:saddle}%
. We consider the buyer surplus generated by the optimal rule (instead of
the efficient social welfare) in \eqref{fff} following some of the
literature. For example, \cite{roge03} and \cite{chsa07} compare the
performance of simple mechanisms relative to the Bayes optimal mechanisms in
specific parametric environments. This approach is also pursued in \cite%
{cadw19} with notable applications to multi-unit optimal pricing and
multi-unit auctions. These are problems where the solution to the optimal
Bayesian problem is either unknown or computationally complex, and hence the
question arises whether a simple solution adapted to the distribution of
uncertainty can attain a good approximation in the sense of the competitive
ratio.

\pagebreak

\bibliographystyle{econometrica}
\bibliography{general}

\end{document}